\documentclass[nocopyrightspace]{sigplanconf}

\usepackage{amsmath}

\usepackage{mathpartir}
\usepackage{listings}
\usepackage[colorinlistoftodos]{todonotes}
\usepackage{url}

\usepackage{paralist}

\usepackage{amssymb}
\usepackage{amsmath}
\usepackage{amsthm}

\usepackage{stmaryrd}

\usepackage{booktabs}

\usepackage{dsfont}

\lstdefinelanguage{JavaScript}{%
  keywords={%
    attributes, class, classend, do, empty, endif, endwhile, fail, function,
    functionend, if, implements, in, inherit, inout, not, of, operations, out,
    return, then, types, while, use, else, switch, case, break, default, for,
    var, with, new, eval
  },
  keywordstyle=\color{black}\bfseries,
  ndkeywords={%
    use, strict
  },
  ndkeywordstyle=\color{black}\bfseries,
  identifierstyle=\color{black},
  sensitive=true,
  comment=[l]{//},
  morecomment = [s]{/*}{*/},
  morecomment = [s][\color{gray}]{/**}{*/},
  commentstyle=\color{gray},
  stringstyle=\color{black},
  basicstyle=\itshape,
}

\newcommand{\TJS}{\textsf{TreatJS}}

\newtheorem{definition}{Definition}

\newtheorem{theorem}{Theorem}
\newtheorem{conjecture}{Conjecture}
\newtheorem{example}{Example}
\newtheorem{lemma}{Lemma}

\newcommand{\bbc}{::=}
\newcommand{\abc}{+\!\!=}

\newcommand{\syntax}{\textit}

\newcommand{\reduce}{\longrightarrow}
\newcommand{\topreduce}{\Longrightarrow}

\newcommand{\subst}[2]{[#1\mapsto#2]}
\newcommand{\dom}[1]{\syntax{dom(}#1\syntax{)}}

\newcommand{\J}{\lambda_\textsf{J}}
\newcommand{\lj}{\J}

\newcommand{\Con}{\lambda_\textsf{Con}}
\newcommand{\lcon}{\Con}

\newcommand{\ljTerm}{\ljM}
\newcommand{\ljConst}{K}
\newcommand{\ljVar}{\ljx}
\newcommand{\ljVal}{\ljV}
\newcommand{\ljOp}{op}
\newcommand{\ljFun}{\lambda{\ljVar}.\ljTerm}
\newcommand{\ljIf}{\syntax{if}}
\newcommand{\ljAt}{\syntax{@}}
\newcommand{\ljEval}{\syntax{eval}}
\newcommand{\ljBlame}{\syntax{B}}

\newcommand{\ljPBlame}{\syntax{+blame}}
\newcommand{\ljNBlame}{\syntax{-blame}}

\newcommand{\ljM}{M}
\newcommand{\ljN}{N}
\newcommand{\ljL}{L}

\newcommand{\ljU}{U}
\newcommand{\ljV}{V}
\newcommand{\ljW}{W}

\newcommand{\ljx}{x}
\newcommand{\ljy}{y}
\newcommand{\ljz}{z}

\newcommand{\ljCtx}{\mathcal{E}}
\newcommand{\ljHole}{\square}

\newcommand{\ctx}[1]{\ljCtx\llbracket#1\rrbracket}

\newcommand{\ljTrue}{\syntax{true}}
\newcommand{\ljFalse}{\syntax{false}}

\newcommand{\defFlatC}[1]{\syntax{flat(}{#1}\syntax{)}}
\newcommand{\defFunC}[2]{#1\rightarrow{#2}}
\newcommand{\defDepC}[1]{\ljVar\rightarrow{#1}}
\newcommand{\defCapC}[2]{{#1}\cap{#2}}
\newcommand{\defCupC}[2]{{#1}\cup{#2}}

\newcommand{\defAbsC}[2]{\Lambda#1.#2}

\newcommand{\conC}{C}
\newcommand{\conD}{D}

\newcommand{\conA}{A}

\newcommand{\conI}{I}
\newcommand{\conJ}{J}

\newcommand{\conQ}{Q}
\newcommand{\conR}{R}

\newcommand{\cbVar}{\iota}
\newcommand{\cbBlame}{\flat}
\newcommand{\cbIdent}{b}

\newcommand{\topassert}[3]{#1\ljAt^{#2}#3}
\newcommand{\assertWith}[4]{#1\ljAt^{#2}_{#3}#4}
\newcommand{\assert}[3]{#1\ljAt^{\cbBlame}_{#2}#3}

\newcommand{\blame}[1]{\ljBlame^{#1}}

\newcommand{\update}{\blacktriangleleft}

\newcommand{\cbState}{\varsigma}
\newcommand{\cbCstr}{\kappa}

\newcommand{\emptystate}{\cdot}

\newcommand{\defCb}[2]{#1\update#2}

\newcommand{\append}[2]{#2\concat#1}
\newcommand{\concat}{:}
\newcommand{\satisfies}[2]{#1\models#2}

\newcommand{\cbTruth}{\mathds{B}}

\newcommand{\cbSol}{\mu}
\newcommand{\Rangeof}[1]{\llparenthesis#1\rrparenthesis}

\newcommand{\subject}[1]{#1.\subjectName}
\newcommand{\context}[1]{#1.\contextName}

\newcommand{\subjectName}{\textit{subject}}
\newcommand{\contextName}{\textit{context}}

\newcommand{\makeTruth}[1]{\tau(#1)}

\newcommand{\unwrap}[1]{\nabla(#1)}

\newcommand{\free}[1]{\syntax{(free?}\,#1\syntax{)}}

\newcommand{\termS}{S}
\newcommand{\termT}{T}

\newcommand{\termSVal}{\termS_{\textsf{Val}}}
\newcommand{\termSNVal}{\termS_{\textsf{NonVal}}}

\newcommand{\termTConst}{\termT_{\textsf{Const}}}
\newcommand{\termTAbs}{\termT_{\textsf{Abs}}}
\newcommand{\termTVal}{\termT_{\textsf{Val}}}

\newcommand{\termTI}{\termT_{\textsf\conI}}
\newcommand{\termTQ}{\termT_{\textsf\conQ}}
\newcommand{\termTNQ}{\termT_{\textsf{Non\conQ}}}
\newcommand{\ljFork}{\parallel} 

\newcommand{\fork}[2]{(#1\ljFork#2)}
\newcommand{\forked}[2]{(#1\ljFork#2)}

\newcommand{\translate}{\longmapsto}

\newcommand{\translateB}{\translate_{\textsf{B}}}

\newcommand{\translateS}{\translate_{\textsf{S}}}
\newcommand{\translateSInT}{\translate_\textsf{S}'}

\newcommand{\join}{\translate_\textsf{J}}
\newcommand{\joinInT}{\translate_\textsf{J}'}

\newcommand{\condense}{\translate_\textsf{C}}

\newcommand{\translateA}{\translate_{\textsf{A}}}
\newcommand{\translateAInT}{\translate_\textsf{A}'}

\newcommand{\optimize}{\Longmapsto} 

\newcommand{\ctxF}{\mathcal{F}}
\newcommand{\ctxG}{\mathcal{G}}
\newcommand{\ctxH}{\mathcal{H}}

\newcommand{\ctxB}{\mathcal{B}}

\newcommand{\inB}[1]{\ctxB\llbracket#1\rrbracket}

\newcommand{\ctxA}{\mathcal{A}}
\newcommand{\ctxV}{\mathcal{V}}

\newcommand{\inA}[1]{\ctxA\llbracket#1\rrbracket}
\newcommand{\inV}[1]{\ctxV\llbracket#1\rrbracket}

\newcommand{\fctx}[1]{\ctxF\llbracket#1\rrbracket}
\newcommand{\gctx}[1]{\ctxG\llbracket#1\rrbracket}

\newcommand{\actx}[1]{\ctxA\llbracket#1\rrbracket}
\newcommand{\bctx}[1]{\ctxB\llbracket#1\rrbracket}

\newcommand{\ctxT}{\mathcal{T}}
\newcommand{\inT}[1]{\ctxT\llbracket#1\rrbracket}

\newcommand{\blameOf}[2]{\eta(#1,\,#2)}

\newcommand{\ConSubsetSubject}[2]{#1\sqsubset^{+}#2}

\newcommand{\ConSubseteq}{\ConOSubseteq}

\newcommand{\ConOSubseteq}[2]{#1\sqsubseteq#2}
\newcommand{\ConNSubseteq}[2]{#1\sqsubseteq^*#2}

\newcommand{\ConSubseteqSubject}[2]{#1\sqsubseteq^{+}#2}
\newcommand{\ConSubseteqContext}[2]{#1\sqsubseteq^{-}#2}

\newcommand{\ConSubsetState}{\Gamma}

\newcommand{\ConSubsetTerm}[2]{#1\leq#2}

\newcommand{\termO}{O}
\newcommand{\termP}{P}

\newcommand{\isEquivCtx}[2]{#1\simeq#2}
\newcommand{\isEquivTerm}[2]{#1\simeq#2}

\newcommand{\inhole}[2]{#1\llbracket#2\rrbracket}

\newcommand{\ctxJoin}[2]{#1\sqcup#2}
\newcommand{\ctxMinus}[2]{#1\setminus#2}
\newcommand{\ctxIn}[2]{#1\in#2}

\newcommand{\ctxNotIn}[2]{#1\not\in#2}

\newcommand{\ctxM}{\mathcal{M}}
\newcommand{\ctxN}{\mathcal{N}}
\newcommand{\ctxL}{\mathcal{L}}

\newcommand{\inM}[1]{\ctxM\llbracket#1\rrbracket}
\newcommand{\inN}[1]{\ctxN\llbracket#1\rrbracket}

\newcommand{\rootof}[2]{\textit{root-of}(#1,\,#2)}
\newcommand{\signof}[2]{\textit{sign-of}(#1,\,#2)}
\newcommand{\inv}[1]{\neg#1}

\newcommand{\ctxK}{\mathcal{K}}
\newcommand{\inK}[1]{\ctxK\llbracket#1\rrbracket}

\newcommand{\extendI}[2]{#1\hookleftarrow#2}

\newcommand{\RuleBeta}{Beta}
\newcommand{\RuleOp}{Op}
\newcommand{\RuleIfTrue}{True}
\newcommand{\RuleIfFalse}{False}

\newcommand{\RuleAssert}{Assert}
\newcommand{\RuleFlat}{Flat}
\newcommand{\RuleUnit}{Unit}

\newcommand{\RuleUnion}{Union}
\newcommand{\RuleIntersection}{Intersection}

\newcommand{\RuleDFunction}{D-Function}
\newcommand{\RuleDDependent}{D-Dependent}
\newcommand{\RuleDIntersection}{D-Intersection}

\newcommand{\RuleBase}{Base}
\newcommand{\RuleDOp}{D-Op}
\newcommand{\RuleDIf}{D-If}

\newcommand{\RuleReduceFalse}{Reduce/False}
\newcommand{\RuleReduceTrue}{Reduce/True}

\newcommand{\RuleDDomain}{D-Domain}
\newcommand{\RuleDRange}{D-Range}
\newcommand{\RuleDFalse}{D-False}

\newcommand{\RuleCSEmpty}{CS-Empty}
\newcommand{\RuleCSState}{CS-State}

\newcommand{\RuleCTIndirection}{CT-Indirection}
\newcommand{\RuleCTInversion}{CT-Inversion}
\newcommand{\RuleCTFlat}{Ct-Flat}
\newcommand{\RuleCTFunction}{CT-Function}
\newcommand{\RuleCTIntersection}{CT-Intersection}
\newcommand{\RuleCTUnion}{Ct-Union}

\newcommand{\RuleUnfoldAssert}{Unfold/Assert}
\newcommand{\RuleUnfoldUnion}{Unfold/Union}
\newcommand{\RuleUnfoldIntersection}{Unfold/Intersection}
\newcommand{\RuleUnfoldOp}{Unfold/Op}

\newcommand{\RuleUnfoldDFunction}{Unfold/D-Function}
\newcommand{\RuleUnfoldDIntersection}{Unfold/D-Intersection}

\newcommand{\RuleConvertTrue}{Convert/True}

\newcommand{\RuleVerifyTrue}{Verify/True}
\newcommand{\RuleVerifyFalse}{Verify/False}

\newcommand{\RuleReverseI}{Push/Immediate}
\newcommand{\RuleReverseFalse}{Push/False}
\newcommand{\RuleReverseIf}{Push/If}

\newcommand{\RuleUnroll}{Unroll}

\newcommand{\RuleLower}{Lower}
\newcommand{\RuleCondense}{Condense}
\newcommand{\RuleBaseline}{Baseline}
\newcommand{\RuleTrace}{Trace}

\newcommand{\RuleLift}{Lift}
\newcommand{\RuleBlame}{Blame} 

\newcommand{\RuleBlameIfTrue}{Blame/If/True}
\newcommand{\RuleBlameIfFalse}{Blame/If/False}
\newcommand{\RuleBlameGlobal}{Blame/Global}

\newcommand{\RuleForkUnion}{Fork/Union}
\newcommand{\RuleForkIntersection}{Fork/Intersection}

\newcommand{\RuleSubsetInner}{Subset/Left}
\newcommand{\RuleSubsetOuter}{Subset/Right}

\newcommand{\RuleMerge}{Merge}

\newcommand{\RuleJoin}{Join}
\newcommand{\RuleJoinUnit}{Match}
\newcommand{\RuleJoinContract}{Synchronize/Contract}
\newcommand{\RuleJoinTermLeft}{Synchronize/Left}
\newcommand{\RuleJoinTermRight}{Synchronize/Right}

\newcommand{\RuleSimplifyUnion}{Simplify/Union}
\newcommand{\RuleSimplifyIntersection}{Simplify/Intersection}

\newcommand{\RuleExtendPredicate}{Extend/Predicate}
\newcommand{\RuleExtendNonPredicate}{Extend/Non-Predicate}

\newcommand{\RuleSubset}{Subset}

\begin{document}

\setlength{\pdfpageheight}{\paperheight}
\setlength{\pdfpagewidth}{\paperwidth}

\conferenceinfo{CONF 'yy}{Month d--d, 20yy, City, ST, Country}
\copyrightyear{20yy}
\copyrightdata{978-1-nnnn-nnnn-n/yy/mm}
\copyrightdoi{nnnnnnn.nnnnnnn}


\preprintfooter{Static Contract Simplification}   

\title{Static Contract Simplification}
\subtitle{Technical Report}

\authorinfo{Matthias Keil \and Peter Thiemann}
           {University of Freiburg,\\Freiburg, Germany}
           {\url{{keilr,thiemann}@informatik.uni-freiburg.de}}

\maketitle


\begin{abstract}

Contracts and contract monitoring are a powerful mechanism for
specifying properties and guaranteeing them at run time.
However, run time monitoring of contracts imposes a
significant overhead. The execution time is impacted by the insertion
of contract checks as well as by the introduction of proxy objects
that perform delayed contract checks on demand.

Static contract simplification attacks this issue using program
transformation. It applies compile-time transformations to programs
with contracts to reduce the overall run time while
preserving the original behavior.
Our key technique is to statically propagate contracts through the
program and to evaluate and merge contracts where possible.
The goal is to obtain residual contracts that are collectively cheaper
to check at run time. 

We distinguish different levels of preservation of behavior, which
impose different limitations on the admissible transformations: Strong
blame preservation, where the transformation is a behavioral
equivalence, and weak blame preservation, where the transformed
program is equivalent up to the particular violation reported.
Our transformations never increase the overall number of
contract checks.
\end{abstract}






\section{Introduction}
\label{sec:introduction}

Software contracts \cite{Meyer1988} 
stipulate invariants for objects as well as pre- and postconditions for
functions a programmer regards as essential for the correct execution of a program.
Traditionally, contracts are checked at run time by contract monitoring.
This approach has become a prominent mechanism to provide strong
guarantees for programs in dynamically typed languages while preserving their
flexibility and expressiveness.

Using formal and precise specifications in the form of contracts in
this way sounds
appealing, but it comes with a cost: Contract monitoring degrades the
execution time significantly
\cite{TakikawaFelteyGreenmanNewVitekFelleisen2016}. This cost arises
because every contract extends a program with additional code to
check the contract assertions that needs to be executed on top of the
program. Moreover, human developers may inadvertently add contracts to
frequently used functions and objects so that contracts end up on hot
paths in a program. In particular, if contracts are naively applied
to recursive functions, 
predicates will check the same values repeatedly, which aggravates the
performance degradation.

Existing contract systems for dynamically typed languages (e.g., Racket's contract framework~\cite[Chapter
7]{FlattFindlerPLT/2014/theracketguide}, Disney's JavaScript contract system
\emph{contracts.js}~\cite{Disney2013:contracts}, or {\TJS} for
JavaScript~\cite{KeilThiemann2015-treatjs}) each suffer from a
considerable performance impact when extending a program with contracts.

In contrast, static contract checking \cite{XuPeytonJonesClaessen2009} avoids
any run-time cost if it succeeds in proving and removing a contract at
compile time. Otherwise, it defers to plain contract checking. A
subsequent extension to hybrid contract checking~\cite{Xu2012} is able
to exploit partial proofs and to extract a residual contract that
still needs to be checked at run time. Both of these works are phrased
as transformations in a typed intermediate language.

A similar motivation is behind static contract verification
\cite{NguyenTobinHochstadtVanHorn2014}. This work takes place in the
context of an untyped language. It applies symbolic execution as well
as techniques from occurrence typing to prove contracts statically and
thus eliminate their run-time overhead.

Our work is complementary to these previous efforts. Our main objective is to
improve the efficiency of contract monitoring by using static
techniques to evaluate as much of a contract as possible, to detect
and remove redundant checking of contracts, and to reassemble the
fragments that cannot be discharged statically in places where they
are least likely to affect performance. Typically, we move contract fragments
to a surrounding module boundary and compose them to a new contract.

Thus, our goal is not to detect contract violations at compile time, but rather to
optimize contracts based on static information. We also do not rewrite the
non-contract code of the underlying program, but perform
program transformations restricted to just moving around contracts.

To demonstrate the essence of our approach, consider the following code snippet
that contains a function with an intersection contract applied to it.
\begin{lstlisting}[name=introduction]
let addOne = 
((,\ plus (,\ z ((plus 1) z))) 
 [(,\ x (,\ y (+ x y)))
  @ ((Number? -> (Number? -> Number?)) cap
        (String? -> (String? -> String?)))])
\end{lstlisting}
The value \lstinline{addOne} is a function that takes an argument
\lstinline{z} and returns the result of applying function
\lstinline{plus} to the number \lstinline{1} and to
\lstinline{z}. We use the inline operator \lstinline{@} to
indicate contract assertions that we put in square
brackets to improve readability. We further rely on predefined flat
contracts \lstinline{Number?} and \lstinline{String?} that are type
checks for numbers and strings, respectively. 

If we assume that the \lstinline{+} operator is overloaded as in
JavaScript or Java so that it works for number and string
values (but possibly fails or delivers meaningless results for other
inputs), a programmer may add an intersection 
contract to \lstinline{plus}. This
contract has the effect that the context of a use of \lstinline{plus}
(here function \lstinline{addOne}) can choose to use \lstinline{plus}
either with number or string arguments. Monitoring this contract requires six predicate 
checks at every use of \lstinline{addOne}.\footnote{%
  Assuming the semantics of intersection and union contracts from the
  literature~\cite{KeilThiemann2015-blame}.}

However, as the second argument of \lstinline{plus} is already a
number,  the right side of the intersection can never be 
fulfilled. So, it would be sufficient to check if the first argument
of \lstinline{plus} and its return value satisfy
\lstinline{Number?}. Moreover, the remaining predicates can be lifted
to an interface description on \lstinline{addOne}, as the following example
demonstrates. Only two predicate checks remain.
\begin{lstlisting}[name=introduction]
let addOne =
[((,\ plus (,\ z ((plus 1) z)))
   ((,\ x (,\ y (+ x y))))) @ (Number? -> Number?)]
\end{lstlisting}
This simplification can be done without knowing whether \lstinline{addOne} is
ever executed: we are not losing any blame due to the
transformation. However, it would not be ok to lift argument contracts
across function boundaries. Such lifting would introduce checks that
may never happen in the original program.


Building on the formalization of \citet{KeilThiemann2015-blame}, we present a
transformation system for static contract simplification. Our system
provides two stages of simplification with different consequences: The \emph{Baseline
Optimization} and the \emph{Subset Optimization}.

The \emph{Baseline Optimization} unfolds contracts and evaluates
predicates where possible while preserving the blame behavior of a program. The
\emph{Subset Optimization} reorganizes contracts and forms new contracts from
the contract fragments, but it may reorder contract computations and
thus change the blame (i.e., the particular violation that is flagged
at run time). Our simplification rules respect three overall guidelines:
\begin{description}
  \item[Improvement]
    Each transformation step improves the efficiency of contract monitoring: it either \emph{reduces}
    or \emph{maintains} the total number of predicate checks at run time.
    Simplification never increases the number of run-time checks.

  \item[Strong Blame-Preservation]
    Each transformation step of the baseline simplification fully preserves the
    blame-behavior of a program. It maintains the oder of contract checks at run
    time and thus the order in which violations are reported. A transformed
    program produces the same outcome as the original program.

  \item[Weak Blame-Preservation]
    Transformation steps of the subset simplification may change the oder of
    predicate checks, so that they may also change the order of observed
    violations. Thus, a transformed program ends in a blame state if and only if
    the original program would end in a blame state. However, the transformed
    program may report a different violation.
\end{description}
An optional third level optimizes contracts across module
boundaries. Even though it provides the strongest optimization, it may
over-approximate contract violations and it would violate all three guideline
from above.

\paragraph{Contributions}
\label{sec:contributions}
This work makes the following contributions:
\begin{itemize}
  \item We present the semantics of a two-tier static contract simplification
    that reorganizes and pre-evaluates contracts at compile time while
    guaranteeing weak or strong blame preservation.
  \item We show how to split alternatives into separated observation to deal
    with intersection and union contracts.    
  \item We give a specification of subcontracting that is used to simplify
    contracts based on the knowledge of other asserted contracts.
\end{itemize}
We created an executable implementation of the semantics using PLT
Redex~\cite{FelleisenFindlerFlatt2009}. The implementation is available online
at \url{http://anonymous}.

\paragraph{Overview}
\label{sec:overview}

Section~\ref{sec:overview} provides a series of examples demonstrating the
essence of our contract simplification. 
Section~\ref{sec:contracts} recalls contracts and contract satisfaction from the
literature and introduces an untyped, applied, call-by-value lambda calculus
with contract monitoring.
Section~\ref{sec:baseline} introduces \emph{Baseline Simplification}, the
first tier in our simplification stack. 
Section~\ref{sec:subset} shows the second tier, \emph{Subset
Simplification}, which also handles alternatives and reduces contracts that are
subsumed by other contracts.
Section~\ref{sec:technical-results} states our blame theorems and
Section~\ref{sec:evaluation} briefly addresses run time improvements. 
Section~\ref{sec:related-work} discusses related work and
Section~\ref{sec:conclusion} concludes.

The appendix extends this paper with auxiliary functions and definitions,
further examples, and proofs of all theorems. In particular, it provides the
semantics to enable picky-like evaluation semantics and the semantics of a
third, approximating, simplification level that provides the stringes 
simplification.


                                             
\section{Contract Simplification}
\label{sec:overview}

This section explains the main ideas of our transformation through a series of
examples. We start with a simple example and work up to more complex ones.

\subsection{Unrolling Delayed Contracts}
\label{sec:overview/unroll}

The examples consider different contracts on the \emph{addOne} example from the
previous section. In a first step we add a simple function contract to
\lstinline{plus}.
\begin{lstlisting}[name=overview]
let addOne = 
((,\ plus (,\ z ((plus 1) z))) 
 [(,\ x (,\ y (+ x y)))
  @ (Number? -> (Number? -> Number?))])
\end{lstlisting} 
While \lstinline{Number?} is a flat contract that is checked immediately,
function contracts are \emph{delayed}
contracts. A delayed contract cannot be checked directly when asserted to a value. Delayed
contracts stay attached to a value and are checked when the value is used.
Thus, every call to \lstinline{addOne} triggers three predicate checks.

Instead of asserting a contract that forces a number of delayed
contract checks, our \emph{Baseline Simplification} unrolls a delayed contract
to all uses of that value. This step pushes the delayed contract
inwards to all occurrences of \lstinline{plus} in \lstinline{addOne}.
\begin{lstlisting}[name=overview]
let addOne = 
((,\ plus
 (,\ z ([plus @ (Number? -> (Number? -> Number?))] 1) z)))
 (,\ x (,\ y (+ x y))))
\end{lstlisting} 
This step performs no simplification, but it yields an equivalent program
that is better prepared for unfolding the function contract to its domain and
range because the contract is on an expression in application position. The
following snippet demonstrates the (complete) unfolding of the contract, which
combines a number of primitive unfolding steps.
\begin{lstlisting}[name=overview]
let addOne = 
((,\ plus
 (,\ z [((plus [1 @ Number?]) [z @ Number?]) @ Number?]))
 (,\ x (,\ y (+ x y))))
\end{lstlisting} 
In the resulting expression, the function contracts are completely
decomposed into flat contract components which are spread all over
the expression.

Observe that, after unfolding, there are several flat contracts applied to values. Such
contracts can be checked statically without changing the outcome. For example,
\lstinline{1} satisfies \lstinline{Number?}. All satisfied contracts can safely
be removed from the program.
\begin{lstlisting}[name=overview]
let addOne = 
((,\ plus
 (,\ z [((plus 1) [z @ Number?]) @ Number?]))
 (,\ x (,\ y (+ x y))))
\end{lstlisting} 
In a last step, we push the remaining contract fragments outwards where possible. This
transformation
creates a new
function contract from that contract. The special contract \lstinline{top}
accepts every value.
\begin{lstlisting}[name=overview]
let addOne = 
[((,\ plus
  (,\ z ((plus 1) [z @ Number?])))
  (,\ x (,\ y (+ x y)))) @ (top -> Number?)]
\end{lstlisting}
Afterwards, the outermost expression carries a function contract that checks the
return value of \lstinline{addOne} and one contract is left in the body of
\lstinline{addOne} to check the argument. Again, this step does not simplify
contracts, but it prepares for further simplification steps.

However, further simplification cannot be done while guaranteeing
strong blame preservation. The \emph{Baseline Simplification} is not allowed to
behave differently than the dynamics would do at run time. 

\subsection{Treating Intersection and Union}
\label{sec:overview/intersection_union}

Our next example considers the full \emph{addOne} example from the introduction,
which we repeat for convenience.
\begin{lstlisting}[name=overview]
let addOne = 
((,\ plus (,\ z ((plus 1) z))) 
 [(,\ x (,\ y (+ x y)))
  @ ((Number? -> (Number? -> Number?)) cap
        (String? -> (String? -> String?)))])
\end{lstlisting} 
Intersections of function contracts are \emph{delayed} contracts, they must
stay attached to a value until the value is used in an application. Moreover,
for an intersection, the context can choose to fulfill the left or the right
side of the intersection. Thus, every call to \lstinline{addOne} triggers six
predicate checks.

Again, \emph{Baseline Simplification} unrolls the intersection contract to all uses
of that value and unfolds the contract to its domain and range when used in an
application. The following snippet demonstrates the unfolding of the
intersection contract.
\begin{lstlisting}[name=overview]
let addOne = 
((,\ plus
 (,\ z [[((plus
             [[1 @2 String?] @1 Number?])
            [[z @2 String?] @1 Number?]) @2 String?] @1 Number?]))
 (,\ x (,\ y (+ x y))))
\end{lstlisting} 
In addition, the static contract monitor maintains additional information in the
background (e.g., which contracts belong to the same side of the intersection)
to connect each contract with the enclosing operation. In the example we
indicate the originating operand of the intersection by \lstinline{@1} and \lstinline{@2}.

After unfolding, the \emph{Baseline Simplification} checks flat contracts
applied to values. Here, \lstinline{1} satisfies \lstinline{Number?} but
violates \lstinline{String?}.
While \lstinline{Number?} can now be removed,
information about a failing contract must remain in the source program. In
general, it is not possible to report failure statically because we do not know
if the code containing the failed contract (in this case, function
\lstinline{addOne}) is ever executed.

However, we do not have to leave the original contract in the
expression. Instead, the transformation replaces it by a ``false
contract'' \lstinline{bot}, whose sole duty is to remember a contract
violation and report it at run time. The next snippet shows
the result after processing satisfied and failing contracts and pushing the
remaining fragments outwards.
\begin{lstlisting}[name=overview]
let addOne = 
[[((,\ plus
   (,\ z ((plus
            [1 @2 bot])
           [[z @2 String?] @1 Number?])))
   (,\ x (,\ y (+ x y)))) @1 (top -> Number?)] @2 (top -> String?)]
\end{lstlisting}
Afterwards, the outermost expression carries two function contracts that check the
return value of \lstinline{addOne}; two contracts are left in the body
of \lstinline{addOne} to check the argument.

\subsection{Splitting Alternatives in Separated Observations}
\label{sec:overview/alternatives}

When remembering the last state in Section\
\ref{sec:overview/intersection_union}, we still check \lstinline{String?} even
though we know that right side of the intersection is never fulfilled. This is
because the \emph{Baseline Simplification} preserves the blame behaviour of a
program and is thus not able to deal with alternatives. However, the next
\emph{Subset Simplification} abandons the strong blame preservation to further
optimize contracts.

Instead of unfolding intersection and union contracts, we now observe both
alternatives in separation. After transforming both alternatives, we join the
remaining contracts or even discard a whole branch if it definitely leads to a
violation.

In our example, we see that \lstinline{1 @ String?} leads to a
context failure of the right operand of the intersection contract. Thus, it is useless to keep the
other fragments of the right operand which definitely yields a context
violation. The following code snippet shows the result.
\begin{lstlisting}[name=overview]
let addOne = 
[((,\ plus
  (,\ z ((plus [1 @2 bot]) [z @1 Number?])))
  (,\ x (,\ y (+ x y)))) @1 (top -> Number?)]
\end{lstlisting}
It remains to check \lstinline{z} and \lstinline{addOne}'s return to be a
number. This are two remaining checks per use of \lstinline{addOne}. 

However, we still need to keep the information on the failing alternative.
This is required to throw an exception if also the other alternative fails. 

\subsection{Lifting Contracts}
\label{sec:overview/lift}

As shown in the example above, contracts on the function's body are reassembled
to a new function contract, whereas contracts on function arguments must remain
to preserve the order of predicate checks. 
Continuing the example, our next step constructs a new contract from the
remaining contracts on arguments.
\begin{lstlisting}[name=overview]
let addOne = 
[[((,\ plus
  (,\ z ((plus [1 @2 bot]) z)))
  (,\ x (,\ y (+ x y)))) @1 (Number? -> top)] @1 (top -> Number?)]
\end{lstlisting}
The \lstinline{addOne} function now contains another function contract that
restricts its domain to number values. As this step puts contracts on arguments
in front, it may change the order in which violations arise.
This lifting to module boundaries only reorganizes existing contracts,
but it helps to reduce contracts that are subsumed by other contracts.

\subsection{Condensing Contracts}
\label{sec:overview/condense}

Continuing the example in Section~\ref{sec:overview/lift}, the outermost
expression carries two function contracts that check the argument and the return
value of \lstinline{addOne}. Moreover, both contracts arise from the same source
contract. 
These function contracts may be condensed to a single function contract on
\lstinline{addOne}.
\begin{lstlisting}[name=overview]
let addOne = 
[((,\ plus
  (,\ z ((plus [1 @2 bot]) z)))
  (,\ x (,\ y (+ x y)))) @1 (Number? -> Number?)]
\end{lstlisting}
This step is only possible if the argument contract, here
\lstinline{Number? -> top}, is in negative position and if the return
contract, \lstinline{top -> Number?} is in positive position.

\subsection{Contract Subsets}
\label{sec:overview/subset}

Splitting intersections and unions into separate observations has another
important advantage: We know that in every branch every contract must be
fulfilled. Thus we can relate contracts to other contracts and use
knowledge about other contracts to reduce redundant checks. To
demonstrate an example, we add a second contract to our running example.
\begin{lstlisting}[name=overview]
let addOne = 
((,\ plus [(,\ z ((plus 1) z)) @ (Positive? -> Positive?)]) 
 [(,\ x (,\ y (+ x y)))
  @ ((Number? -> (Number? -> Number?)) cap
        (String? -> (String? -> String?)))])
\end{lstlisting} 
Here, \lstinline{Positive?} is a flat contract that checks for positive numbers.
In addition to the contract on \lstinline{plus}, \lstinline{addOne} contains
another function contract, requiring \lstinline{addOne} to be called with
positive numbers and to return a positive number. After performing the
same transformation steps as before, we obtain the following expressions.

\begin{lstlisting}[name=overview]
let addOne =
[[[((,\ plus (,\ z ((plus 1) z)))
   (,\ x (,\ y (+ x y)))) @1 (Number? -> top)] @1 (top -> Number?)] @ (Positive? -> Positive?)]
||
let addOne =
[((,\ plus (,\ z ((plus [1 @2 bot]) z)))
 (,\ x (,\ y (+ x y)))) @ (Positive? -> Positive?)]
\end{lstlisting}
The intersection is split into two parallel observations, as indicated by
\lstinline{||}. All contracts were pushed to the outermost boundary and form a
new module boundary. However, some of the contract checks are redundant.

For example, the innermost contract in the first observation requires
\lstinline{addOne} to be called with numbers, whereas the outermost contract
restricts the argument already to positive numbers. As positive numbers are a
proper subset of numbers, the inner contract will never raise a violation if the
outer contract is satisfied. Thus, it can be removed without changing the blame
behavior of a program.

The contract \lstinline{Positive? -> Positive?} is more restrictive than any
other contract on \lstinline{addOne}, that is, \lstinline{Number? -> top} and
\lstinline{top -> Number?}. Here, a contract $A$ is more restrictive
than a contract $B$, if the satisfaction of $A$ implies the satisfaction of $B$.

However, if a more restrictive contract fails, that failure does not
imply failure of the less restrictive contract. But as  every contract must be fulfilled, the program
already stops with a contract violation. Thus, there is no need to check a less
restrictive contract on that branch. Unfortunately, removing the less restrictive contract might
change the order of observed violations, as the less restrictive contract might
report its violation first.

To sum up, in the first observation the outer contract remains as it subsumes any
other contract, whereas in the second observation both contracts remain.

Finally, after finishing all transformation steps, the simplifications joins the
observation to final expression.
\begin{lstlisting}[name=overview]
let addOne = 
[((,\ plus
 (,\ z ((plus [1 @2 bot]) z)))
 (,\ x (,\ y (+ x y)))) @ (Positive? -> Positive?)]
\end{lstlisting}

%
%
%
%
                                                        
\section{Contracts and Contract Satisfaction}
\label{sec:contracts}

This section defines $\lcon$, an untyped call-by-value lambda calculus with
contracts which serves as a core calculus for contract monitoring. It first
introduces the base calculus $\J$, then it proceeds to describe contracts and
their semantics for the base calculus. Finally it gives the semantics of
contract assertion and blame propagation.

The core calculus is inspired by previous work from the
literature~\cite{KeilThiemann2015-blame}.


\subsection{The Base Language $\J$}
\label{sec:contracts/lj}

\begin{figure}[tp]
  \begin{displaymath}
    \begin{array}{lrl}
      \ljL, \ljM, \ljN &\bbc& \ljConst \mid \ljVar \mid \ljFun \mid \ljM\,\ljN
      \mid \ljOp(\vec\ljTerm) \mid \ljIf\,\ljL\,\ljM\,\ljN\\
      \ljConst &\bbc& \ljFalse \mid \ljTrue \mid 0 \mid 1 \mid \dots \\
      \ljU, \ljV, \ljW &\bbc& \ljConst \mid \ljFun \\
      \ljCtx &\bbc& \ljHole \mid \ljCtx\,\ljN \mid \ljV\,\ljCtx \mid
      \ljOp(\vec\ljVal\,\ljCtx\,\vec\ljTerm) \mid \ljIf\,\ljCtx\,\ljM\,\ljN
    \end{array}
  \end{displaymath}

  \begin{displaymath}
    \begin{array}{llcll}
      \DefTirName{\RuleBeta} 
      &\ctx{\ljFun\,\ljVal}
      &\reduce
      &\ctx{\ljM\subst\ljVar\ljVal}
      &
      \\
      \DefTirName{\RuleOp}
      &\ctx{\ljOp (\vec\ljVal)}
      &\reduce
      &\ctx{\delta_{\ljOp}(\vec\ljVal)}
      &\vec{\ljVal} \in \dom{\delta_{\ljOp}}
      \\
      \DefTirName{\RuleIfTrue}
      &\ctx{\ljIf\,\ljTrue\,\ljM\,\ljN}
      &\reduce
      &\ctx{\ljM}
      &
      \\
      \DefTirName{\RuleIfFalse}
      &\ctx{\ljIf\,\ljFalse\,\ljM\,\ljN}
      &\reduce
      &\ctx{\ljN}
      &
    \end{array}
  \end{displaymath}
  \caption{%
    Syntax and semantics of $\J$.
  }\label{fig:lj}
\end{figure}

Figure~\ref{fig:lj} defines syntax and semantics of $\lj$, an applied
call-by-value lambda calculus. A term $\ljTerm$ is either a first-order
constant, a variable, a lambda abstraction, an application, a primitive
operation, or a condition. Variables $\ljx$, $\ljy$, and $\ljz$ are drawn from a
denumerable set of variables. Constants $\ljConst$ range over a set of base type
values including booleans, numbers, and strings.

To define evaluation, values $\ljU$, $\ljV$, and $\ljW$ range over values and
$\ljCtx$ over evaluation contexts, which are defined as usual.

The small-step reduction relation $\reduce$ comprises beta-value reduction,
built-in partial operations transforming a vector of values into a value, and
the reduction of conditionals. We write $\reduce^*$ for its reflexive,
transitive closure and $\not\reduce^*$ for its complement. That is,
$\ljM \not\reduce^* \ljN$ if, for all $\ljL$ such that $\ljM \reduce^* \ljL$, it
holds that $\ljL\ne\ljN$. We also write $\ljM \not\reduce$ to indicate that
there is no $\ljN$ such that $\ljM\reduce\ljN$.

                                   
\subsection{Contracts and Contracted $\J$}
\label{sec:contracts/lcon}

\begin{figure}[tp]
  \begin{displaymath}
    \begin{array}{lrl}
      \conC, \conD &\bbc& \conI \mid \conQ \mid \defCupC\conC\conD \mid
      \defCapC\conI\conC\\
      \conI, \conJ &\bbc& \defFlatC{\ljM}\\
      \conQ, \conR &\bbc& \defFunC\conC\conD \mid \defDepC\conA \mid
      \defCapC\conQ\conR\\
      \conA &\bbc& \defAbsC\ljVar\conC\\
    \end{array}
  \end{displaymath}
  \begin{displaymath}
    \begin{array}{lrl}
      \ljBlame &\bbc& \ljPBlame \mid \ljNBlame\\
      \ljL, \ljM, \ljN &\abc& \topassert\ljM\cbBlame\conC \mid\mid
      \assert\ljVal\cbVar{(\ljEval\,\ljM)} \mid \blame\cbBlame\\
      \ljU, \ljV, \ljW &\abc& \mid\mid \assert\ljVal\cbVar\conQ\\ 
      \ljCtx &\abc& \topassert\ljCtx\cbBlame\conC \mid\mid
      \assert\ljVal\cbVar{(\ljEval\,\ljCtx)}
    \end{array}
  \end{displaymath}
  \begin{displaymath}
    \begin{array}{lrl}
      \cbIdent &\bbc& \cbBlame \mid\mid \cbVar\\
      \cbCstr &\bbc& \defCb\cbIdent\ljW \mid \defCb\cbIdent\cbVar \mid \defCb\cbIdent\defFunC\cbVar\cbVar \mid
      \defCb\cbIdent\defCapC\cbVar\cbVar \mid \defCb\cbIdent\defCupC\cbVar\cbVar\\
      \cbState &\bbc& \emptystate \mid \append\cbState\cbCstr
    \end{array}
  \end{displaymath}
  \caption{%
    Syntax extension for $\Con$.
  }\label{fig:lcon}
\end{figure}

Figure~\ref{fig:lcon} defines the syntax of $\Con$ as an extension of $\J$. It
first introduces constructs for contract monitoring in general and, second, it
adds new terms specific to contract monitoring. Intermediate terms that
do not occur in source programs appear after double bars ``$\mid\mid$''.

A contract $\conC$ is either an immediate contact $\conI$, a delayed contract
$\conQ$, a union between two contract $(\defCupC\conC\conD)$, or an intersection
$(\defCapC\conI\conC)$.

Without loss of generality we restrict top-level intersection contracts to
intersections between an immediate contract and a rest contract (cf.\
\cite{KeilThiemann2015-blame}). This requires that unions and immediate parts are
distributivity pulled out of intersections of delayed contracts such that the
immediate parts can be check right away when asserted, whereas the delayed parts
remain.

Immediate contracts $\conI$ and $\conJ$ stand for flat contracts define by a
predicate $\ljM$ that can be evaluated right away when asserted to a value.

A delayed contracts $\conQ$ is either a function contract $\defFunC\conC\conD$
with domain contract $\conD$ and range contract $\conD$, a dependent contract
$\defDepC\conA$ defined by contract abstraction $\conA$, or the finite
intersection of function (dependent) contracts. Delayed contracts stay
with the value until the value is subject of an application.

Consequently, values are extended with $\assert\ljVal\cbVar\conQ$ which
represents a value wrapped in a delayed contract that is to be monitored when
the value is used. The extended set of values forces us to revisit the built-in
operations. We posit that each partial function $\delta_{\ljOp}$ first erases
all contract monitoring from its arguments, then processes the underlying
$\J$-values, and finally returns a $\J$-value.

The only new source term is contract assertion $\topassert\ljM{\cbIdent}\conC$.
The adornments $\cbBlame$ and $\cbVar$ are drawn from an unspecified denumerable
set of blame identifiers $\cbIdent$, which comprises blame labels $\cbBlame$
that occur in source terms and blame variables $\cbVar$ that are introduced
during evaluation.

In the intermediate term $\assert\ljVal\cbVar{(\ljEval\,\ljTerm)}$, the term
$\ljTerm$ represents the current evaluation state of the predicate of a flat
contract. The $\blame\cbBlame$ expression signal a contract violation at label
$\cbBlame$. Term $\ljBlame$ is either a $\ljPBlame$ (positive) or a $\ljNBlame$
(negative) blame, representing subject or context blame, respectively.

Evaluation contexts are extended in the obvious way: a contract
monitor is only applied to a value and a flat contract is checked
before its value is used.

In $\Con$, contract monitoring occurs via constraints $\cbCstr$ imposed on blame
identifiers. There is an indirection constraint and one kind of constraint for
each kind of contract: flat, function, intersection, and union. Constraints are
collected in a list $\cbState$ during reduction.

                                     
\subsection{Semantics for $\Con$}
\label{sec:lcon/semantics}

\begin{figure*}[tp]
  \begin{displaymath}
    \begin{array}{llllllr}

      \DefTirName{\RuleAssert}
      & \cbState,
      & \ctx{\topassert\ljVal\cbBlame\conC}
      & \reduce
      & \append\cbState{(\defCb{\cbBlame}{\cbVar})},
      & \ctx{\assert\ljVal\cbVar\conC}
      & \cbVar\not\in{\cbState}
      \\
      \\

      \DefTirName{\RuleFlat}
      & \cbState,
      & \ctx{\assert\ljVal\cbVar{\defFlatC{\ljM}}}
      & \reduce
      & \cbState,
      & \ctx{\assert\ljVal\cbVar{(\ljEval\,(\ljM\,\ljU))}}
      & \ljU=\unwrap\ljVal
      \\
      \DefTirName{\RuleUnit}
      & \cbState,
      & \ctx{\assert\ljV\cbVar{(\ljEval\,\ljW)}}
      & \reduce
      & \append\cbState{(\defCb{\cbVar}{\ljW})},
      & \ctx{\ljV}
      & 
      \\
      \\

      \DefTirName{\RuleUnion}
      & \cbState,
      & \ctx{\assert\ljVal\cbVar{(\defCupC\conC\conD)}}
      & \reduce
      & \append\cbState{(\defCb\cbVar\defCupC{\cbVar_1}{\cbVar_2})},
      & \ctx{\assert{(\assert\ljVal{\cbVar_1}\conC)}{\cbVar_2}\conD}
      & \cbVar_1,\cbVar_2\not\in{\cbState}
      \\

      \DefTirName{\RuleIntersection}
      & \cbState,
      & \ctx{\assert\ljVal\cbVar{(\defCapC\conI\conC)}}
      & \reduce
      & \append\cbState{(\defCb\cbVar\defCapC{\cbVar_1}{\cbVar_2})},
      & \ctx{\assert{(\assert\ljVal{\cbVar_1}\conI)}{\cbVar_2}\conC}
      & \cbVar_1,\cbVar_2\not\in{\cbState}
      \\
      \\

      \DefTirName{\RuleDFunction}
      & \cbState,
      & \ctx{(\assert\ljV\cbVar{(\defFunC\conC\conD)})\, \ljW}
      & \reduce
      & \append\cbState{(\defCb\cbVar\defFunC{\cbVar_1}{\cbVar_2})},
      & \ctx{(\assert{\ljV\, (\assert\ljW{\cbVar_1}\conC))}{\cbVar_2}\conD}
      & \cbVar_1,\cbVar_2\not\in{\cbState}
      \\

      \DefTirName{\RuleDDependent}
      & \cbState,
      & \ctx{(\assert\ljV\cbVar{(\defDepC{\defAbsC\ljVar\conC})})\, \ljW}
      & \reduce
      & \cbState,
      & \ctx{(\assert{\ljV\, \ljW)}{\cbVar}(\conC\subst\ljVar\ljU)}
      & \ljU=\unwrap\ljW 
      \\

      \DefTirName{\RuleDIntersection}
      & \cbState,
      & \ctx{(\assert\ljV\cbVar{(\defCapC\conQ\conR)})\,\ljW}
      & \reduce
      & \append\cbState{(\defCb\cbVar\defCapC{\cbVar_1}{\cbVar_2})},
      & \ctx{(\assert{(\assert\ljV{\cbVar_1}\conQ)}{\cbVar_2}\conR)\, \ljW}
      & \cbVar_1,\cbVar_2\not\in{\cbState}

    \end{array}
  \end{displaymath}

  \begin{mathpar}
    \inferrule[\RuleBase]{%
      \ljM \reduce \ljN
    }{%
      \cbState,\ljM \reduce \cbState,\ljN
    }

    \inferrule[\RuleDOp]{%
      \ljU = \ljConst \mid  \ljFun
    }{%
      \cbState,\ctx{\ljOp(\vec\ljU,\,(\assert\ljV{\cbVar}\conQ),\,\vec\ljW)} \reduce
      \cbState,\ctx{\ljOp(\vec\ljU,\,\ljV,\,\vec\ljW)}
    }

    \inferrule[\RuleDIf]{%
    }{%
      \cbState,\ctx{\ljIf\,(\assert\ljV{\cbVar}\conQ)\,\ljM\,\ljN} \reduce
      \cbState,\ctx{\ljIf\,\ljV\,\ljM\,\ljN}
    }
  \end{mathpar}

  \caption{Operational semantics of $\lcon$.}
  \label{fig:lcon/semantics}
\end{figure*}

\begin{figure}[tp]
  \begin{displaymath}
    \unwrap{\ljV} =
    \begin{cases}
      \unwrap{\ljW} & \ljV = \assert\ljW\cbVar\conQ \\
      \ljV & \text{otherwise}
    \end{cases}
  \end{displaymath}
  \caption{Unwrap delayed contracts.}
  \label{fig:lcon/unwrap}
\end{figure}

Figure~\ref{fig:lcon/semantics} specifies the small-step reduction semantics of
$\Con$ as a relation $\cbState, \ljM \reduce \cbState', \ljN$ on pairs of a
constraint list and a term. Instead of raising blame exceptions, the rewriting
rules for contract enforcement generate constraints in $\cbState$: a failing
contract must not raise blame immediately, because it may be nested in an
intersection or a union. The sequence of elements in the constraint list
reflects the temporal order in which the constraints were generated during
reduction. The latest, youngest constraints are always on top of the list.
Appendix~\ref{sec:appendix/constraints} explains the semantics of the
constraints.

The rule \RefTirName{\RuleAssert} introduces a fresh name for each new
instantiation of a monitor in the source program. It is needed for technical
reasons. Flat contracts get evaluated immediately. Rule \RefTirName{\RuleFlat}
starts evaluating a flat contract by evaluating the predicate $\ljM$ applied to
the unwrapped subject value $\ljW$. Meta function $\unwrap\ljV$ (Figure
\ref{fig:lcon/unwrap}) first erases all contract monitors from the subject
value $\ljV$. After predicate evaluation, rule \RefTirName{\RuleUnit} picks up
the result and stores it in a constraint.

The second group of rules \RefTirName{\RuleUnion} and
\RefTirName{\RuleIntersection} implements contract checking of top-level union
and intersection contracts. Both rules install a new constraint that combine the
contract satisfaction of the subcontracts to the satisfaction of the current
contract.

The rules \RefTirName{\RuleDFunction}, \RefTirName{\RuleDDependent},
and \RefTirName{\RuleDIntersection} define the behavior of a contracted value under
function application. Rule \RefTirName{\RuleDFunction} handles a call to a value
with a function contract. Different from previous work, the blame computation is
handled indirectly by creating new blame variables for the domain and range
part; a new constraint is added that transforms the outcome of both portions
according to the specification of the function contract.

The \RefTirName{\RuleDDependent} rule defines the evaluation of a dependent
contract. It substitutes the function's arguments for every occurrence of
$\ljVar$ in the body of $\conC$ and proceeds with the assertion of the resulting
contract.

Rule \RefTirName{\RuleDIntersection} embodies the evaluation of an intersection
contract under function application. The generated constraint serves to combine
the results of the subcontracts. Unlike the union contract, an intersection
constraint occurs at each use of the contracted value, which implements the
choice of the context (cf.\ \cite{KeilThiemann2015-blame}).

Finally, $\J$ reductions are lifted to $\Con$ reductions using rule
\RefTirName{\RuleBase}. Built-in operations and conditions ``see through''
contracts (rules \RefTirName{\RuleDOp} and \RefTirName{\RuleDIf}) and proceed on
$\J$-values without contracts. This also implies that all base type operations
only return values that does not contain contracts.


\subsection{Introducing Blame}
\label{sec:lcon/blame}

The dynamics in Figure~\ref{fig:lcon/semantics} use constraints to create a
structure for computing positive and negative blame according to the semantics
of subject and context satisfaction, respectively. To this end, each blame
identifier $\cbIdent$ is associated with two truth values, $\subject{\cbIdent}$
and $\context{\cbIdent}$. Intuitively, if $\subject{\cbIdent}$ is false, then
the contract associated with $\cbIdent$ is not subject-satisfied and may lead to
positive blame for $\cbIdent$. If $\context\cbIdent$ is false, then there is a
context that does not respect contract $\cbIdent$ and may lead to negative blame
for $\cbIdent$. Appendix~\ref{sec:appendix/constraints} specifies the constraint
satisfaction for $\Con$.

Computing a blame assignment boils down to computing an interpretation for a
constraint list $\cbState$. To determine whether a constraint list $\cbState$ is
a blame state (i.e., whether it should signal a contract violation), we check
whether the state $\cbState$ maps any source-level blame label $\cbBlame$ to
$\ljFalse$.

\begin{definition}\label{def:blame-state}
  $\cbState$ is a \emph{blame state for blame label $\cbBlame$} iff
  $$
  \exists\cbBlame.\cbState(\subject\cbBlame \wedge \context\cbBlame) \sqsupseteq
  \ljFalse.
  $$
  $\cbState$ is a \emph{blame state} if there exists a blame label
  $\cbBlame$ such that $\cbState$ is a blame state for this label.
\end{definition}

To model reduction with blame, we define a new reduction relation $\cbState,
\ljM \topreduce \cbState', \ljN$ on configurations. It behaves like $\reduce$
unless $\cbState$ is a blame state. In a blame state, it stops signaling the
violation. There are no reductions with $\ljBlame$.

\begin{mathpar}
  \infer{%
    \exists\cbBlame.\cbState(\subject\cbBlame) \sqsupseteq \ljFalse
  }{%
    \cbState,\ljM \topreduce \cbState,\ljPBlame^\cbBlame
  }

  \infer{%
    \exists\cbBlame.\cbState(\context\cbBlame) \sqsupseteq \ljFalse
  }{%
    \cbState,\ljM \topreduce \cbState,\ljNBlame^\cbBlame
  }

  \infer{%
    \cbState,\ljM \reduce \cbState',\ljN\\\\
    \not\exists\cbBlame.\cbState(\subject\cbBlame \wedge \context\cbBlame) \sqsupseteq
    \ljFalse
  }{%
    \cbState,\ljM \topreduce \cbState',\ljN
  }
\end{mathpar}

                                     
\subsection{Lax, Picky, and Indy Semantics}
\label{sec:lcon/semantics}

Contract monitoring distinguishes three blame semantics: \emph{Lax},
\emph{Picky}, and \emph{Indy}, initially introduced to handle correctness and
completeness of higher-order dependent contracts\
\cite{DimoulasFindlerFlanaganFelleisen2011}. The different styles address the
point that a contract might violate another contract.

The \emph{Lax} semantics erase all contract monitors on values before applying
a value to a predicate. This is correct but not complete because it guarantees
that a well-behaved program gets never blamed for a violation taking place in a
predicate. However, it omits to report violations in predicates.

The \emph{Picky} semantics, in contrast, is complete but not correct. As it
retains contract monitors on the values it might wrongly blame the program for
violating a contract, but it guarantees that every violation is reported.

The \emph{Indy} semantics, an extension of \emph{Picky}, introduces a third
party (the contract itself) in addition to the exiting parties context and
subject. Passing a value to a predicate reorganizes the monitor such that 
wrong uses of that value blame the ill-behaved contract. 

Our optimization considers the \emph{Lax} style, as it strongly guarantees
correctness and completeness in respect to source program. As reorganizing
contracts changes the order of predicate checks optimizing is simpler when
using \emph{Lax} semantics.

Appendix~\ref{sec:appendix/picky} discusses our simplification with respect to
\emph{Picky} like evaluation semantics.
 
%
%

\section{Baseline Simplification}
\label{sec:baseline}

This section presents the \emph{Baseline Simplification}, the first tier in our
simplification stack. The baseline simplification encompass \begin{inparaenum}[(i)]
\item the unfolding of contract assertions and the creation of constraints where
  possible,
\item the compile-time evaluation of flat contracts (predicates) on values, and
\item the unrolling of delayed contracts to all uses of the contracted values.
\end{inparaenum} To this end, we introduce a set of canonical
(non-transformable) terms $\termT$ as a subset of $\Con$ terms $\ljM$ and extend
syntax and semantics of $\Con$ in some respects.

All baseline steps do exactly the same as the dynamics would later on do.
Thus, baseline simplification guarantees \emph{Strong Blame-preservation} and
\emph{Improvement} for all transformation steps.

                                                                          
\subsection{Syntax and Semantics Extension}
\label{sec:baseline/extension}

\begin{figure}[tp]
  \begin{displaymath}
    \begin{array}{lrl}
      \conC, \conD &\abc& \top \mid \bot\\
      \ljL, \ljM, \ljN &\abc& \assert\ljTerm\cbVar\conC\\
      \cbCstr &\abc& \defCb\cbIdent{\neg\cbVar} \\
    \end{array}
  \end{displaymath} 
  \caption{%
    Syntax extension of $\Con$.
  }\label{fig:lcon/syntax/extension}
\end{figure}

\begin{figure*}[tp]
  \begin{displaymath}
    \begin{array}{llllllr}

      \DefTirName{\RuleReduceTrue}
      & \cbState,
      & \ctx{\assert\ljV\cbVar\top}
      & \reduce
      & \append\cbState{(\defCb{\cbVar}{\ljTrue})},
      & \ctx{\ljV}
      &
      \\

      \DefTirName{\RuleReduceFalse}
      & \cbState,
      & \ctx{\assert\ljV\cbVar\bot}
      & \reduce
      & \append\cbState{(\defCb{\cbVar}{\ljFalse})},
      & \ctx{\ljV}
      &
      \\
      \\

      \DefTirName{\RuleDDomain}
      & \cbState,
      & \ctx{(\assert\ljV\cbVar{(\defFunC\conC\top)})\, \ljW}
      & \reduce
      & \append\cbState{(\defCb\cbVar{\neg\cbVar_1})},
      & \ctx{\ljV\, (\assert\ljW{\cbVar_1}\conC)}
      & \cbVar_1\not\in{\cbState}
      \\
      \DefTirName{\RuleDRange}
      & \cbState,
      & \ctx{(\assert\ljV\cbVar{(\defFunC\top\conD)})\, \ljW}
      & \reduce
      & \append\cbState{(\defCb\cbVar{\cbVar_1})},
      & \ctx{(\assert{\ljV\, \ljW)}{\cbVar_1}\conD}
      & \cbVar_1\not\in{\cbState} 
      \\
      \DefTirName{\RuleDFalse}
      & \cbState,
      & \ctx{(\assert\ljV\cbVar{(\defFunC\top\bot)})\, \ljW}
      & \reduce
      & \append\cbState{(\defCb{\cbVar}{\ljFalse})},
      & \ctx{\ljV\, \ljW}
      & 
      \\

    \end{array}
  \end{displaymath}
  \caption{Operational semantics extension of $\lcon$.}
  \label{fig:lcon/semantics/extension}
\end{figure*}

Figure~\ref{fig:lcon/syntax/extension} shows some syntax extension for $\Con$.

Contracts now contain $\top$ (true) and $\bot$ (false), shortcuts for predicates
evaluating to true or false, respectively. $\top$ and $\bot$ arise during
optimization and represent remaining knowledge about the given proposition.
Whereas $\top$ can be reduced statically, $\bot$ must remain in the source
program to ensure correct blame propagation.

The shortcuts are more efficient because none of them must be checked at
run time. When reaching a term with $\top$ or $\bot$ the outcome can immediately
be transfered into a constraint and must not be checked like a predicate.

Source terms $\ljM$ now contain partially evaluated contracts on source terms
$\assert\ljTerm\cbVar\conC$ and a new inversion constraint
$\defCb\cbIdent{\neg\cbVar}$ extends the set of constrains $\cbCstr$. The
inversion constraint flips context and subject responsibility
and corresponds to the negative side of a function constraint. This constraint
only arises during optimization.

Figure~\ref{fig:lcon/semantics/extension} specified the reduction of the
shortcuts. Rule\ \RefTirName{\RuleReduceTrue} and\ \RefTirName{\RuleReduceFalse}
creates a constraint for $\top$ and $\bot$. Three special rules,
\RefTirName{\RuleDDomain}, \RefTirName{\RuleDRange}, and
\RefTirName{\RuleDFalse} simplify the evaluation of function contracts with
$\top$ and $\bot$ on its domain or range portion. If one portion is true only
the other portion must be checked. All this function contrast arise only during
optimization.


\subsection{Canonical Terms and Contexts}
\label{sec:baseline/terms}

\begin{figure}[tp]
  \begin{displaymath}
    \begin{array}{lrl}

      \ctxF, \ctxG, \ctxH &\bbc& \ljHole \mid \lambda\ljVar.\ctxF \mid \ctxF\,\ljM \mid \termT\,\ctxF \mid
      \ljOp(\vec\termT\,\ctxF\,\vec\ljM)\\
      &\mid& \ljIf\,\ctxF\,\ljM\,\ljN \mid \ljIf\,\termT\,\ctxF\,\ljN \mid \ljIf\,\termT\,\termT\,\ctxF
      \mid \assert\ctxF\cbVar\conC\\
      \\

      \ctxV &\bbc& \ljHole \mid \assert\ctxV\cbVar\bot\\

    \end{array}
  \end{displaymath}

  \caption{%
    Transformation contexts.
  }\label{fig:baseline/syntax}
\end{figure}

Canonical (non-transformable) terms are the output of our transformation. A
canonical terms is either a contract-free term or a term with a contract that
cannot further transformed at this level.

In general, canonical terms distinguish terms without a contract on the
outermost level (indicated by letter $\termS$) from terms with a
non-transformable contract (indicated by letter $\termT$). Canonical terms are a
subset of $\Con$ terms $\ljM$. Appendix~\ref{sec:appendix/terms/baseline} shows
its definition.

A transformation context $\ctxF$, $\ctxG$, and $\ctxH$
(Figure~\ref{fig:baseline/syntax}) is defined as usual as a term with a hole and
a value context $\ctxV$ is a term with an indefinite number of remaining $\bot$
contracts. This terms result only during transformation.

%
                                                             
\subsection{Baseline Transformation}
\label{sec:baseline/transformation}

The baseline transformation $\cbState,\ljM\translateB\cbState',\ljN'$ is defined
on pairs of a constrains list and terms. As before, we write $\translateB^*$ for
its reflexive, transitive closure.

\begin{figure*}

  \begin{displaymath}
    \begin{array}{llllllr}

      \DefTirName{\RuleUnfoldAssert}
      & \cbState,
      & \fctx{\topassert\termT\cbBlame\conC}
      & \translateB
      & \append\cbState{(\defCb{\cbBlame}{\cbVar})},
      & \fctx{\assert\termT\cbVar\conC}
      & \cbVar\not\in{\cbState}
      \\
      \\

      \DefTirName{\RuleUnfoldUnion}
      & \cbState,
      & \fctx{\assert\termT\cbVar{(\defCupC\conC\conD)}}
      & \translateB
      & \append\cbState{(\defCb\cbVar\defCupC{\cbVar_1}{\cbVar_2})},
      & \fctx{\assert{(\assert\termT{\cbVar_1}\conC)}{\cbVar_2}\conD}
      & \cbVar_1,\cbVar_2\not\in{\cbState}
      \\

      \DefTirName{\RuleUnfoldIntersection}
      & \cbState,
      & \fctx{\assert\termT\cbVar{(\defCapC\conI\conC)}}
      & \translateB
      & \append\cbState{(\defCb\cbVar\defCapC{\cbVar_1}{\cbVar_2})},
      & \fctx{\assert{(\assert\termT{\cbVar_1}\conI)}{\cbVar_2}\conC}
      & \cbVar_1,\cbVar_2\not\in{\cbState}
      \\

      \DefTirName{\RuleUnfoldOp}
      & \cbState,
      & \fctx{\ljOp(\vec\termTI,\,(\assert\termT{\cbVar}\conQ),\,\vec\termTQ)}
      & \translateB
      & \append\cbState{(\defCb{\cbVar}{\ljTrue})},
      & \fctx{\ljOp(\vec\termTI,\,\termT,\,\vec\termTQ)}
      &
      \\
      \\

      \DefTirName{\RuleUnfoldDFunction}
      & \cbState,
      & \fctx{(\assert{\termT_1}\cbVar{(\defFunC\conC\conD)})\, \termT_2}
      & \translateB
      & \append\cbState{(\defCb\cbVar\defFunC{\cbVar_1}{\cbVar_2})},
      & \fctx{(\assert{\termT_1\, (\assert{\termT_2}{\cbVar_1}\conC))}{\cbVar_2}\conD}
      & \cbVar_1,\cbVar_2\not\in{\cbState}
      \\
      
      \DefTirName{\RuleUnfoldDIntersection}
      & \cbState,
      & \fctx{(\assert{\termT_1}\cbVar{(\defCapC\conC\conD)})\,\termT_2}
      & \translateB
      & \append\cbState{(\defCb\cbVar\defCapC{\cbVar_1}{\cbVar_2})},
      & \fctx{(\assert{(\assert{\termT_1}{\cbVar_1}\conC)}{\cbVar_2}\conD)\,\termT_2)}
      & \cbVar_1,\cbVar_2\not\in{\cbState}
      \\
      \\

      \DefTirName{\RuleUnroll}
      & \cbState,
      & \fctx{\inV{\lambda\ljVar.\termS}\,(\assert\termT\cbVar\conQ)}
      & \translateB
      & \cbState
      & \fctx{\inV{\lambda\ljVar.\termS[\ljVar\mapsto(\assert\ljVar\cbVar\conQ)]}\,\termT} 
      & 
      \\

      \DefTirName{\RuleLower}
      & \cbState,
      & \fctx{\lambda\ljVar.(\assert\termT\cbVar\conC)}
      & \translateB
      & \cbState
      & \fctx{\assert{(\lambda\ljVar.\termT)}\cbVar{(\defFunC\top\conC)}}
      & 
      \\
      \\

      \DefTirName{\RuleReverseI}
      & \cbState,
      & \fctx{\assertWith{(\assertWith\termT{\cbBlame_1}{\cbVar_1}\conQ)}{\cbBlame_2}{\cbVar_2}\conI}
      & \translateB
      & \cbState
      & \fctx{\assertWith{(\assertWith\termT{\cbBlame_2}{\cbVar_2}\conI)}{\cbBlame_1}{\cbVar_1}\conQ}
      & 
      \\
      \DefTirName{\RuleReverseFalse}
      & \cbState,
      & \fctx{\assertWith{(\assertWith\termT{\cbBlame_1}{\cbVar_1}\conQ)}{\cbBlame_2}{\cbVar_2}\perp}
      & \translateB
      & \cbState
      & \fctx{\assertWith{(\assertWith\termT{\cbBlame_2}{\cbVar_2}\perp)}{\cbBlame_1}{\cbVar_1}\conQ}
      & 
      \\
      \DefTirName{\RuleReverseIf}
      & \cbState,
      & \fctx{(\ljIf\,\termT_0\,(\assert{\termT_1}{\cbVar}{\conC})\,(\assert{\termT_2}{\cbVar}{\conC}))}
      & \translateB
      & \cbState
      & \fctx{\assert{(\ljIf\,\termT_0\,\termT_1\,\termT_2)}{\cbVar}{\conC}}
      & 
      \\
      \\    

      \DefTirName{\RuleConvertTrue}
      & \cbState,
      & \fctx{\assert\termT\cbVar\top}
      & \translateB
      & \append\cbState{(\defCb{\cbVar}{\ljTrue})},
      & \fctx{\termT}
      &
      \\

    \end{array}
  \end{displaymath}

\begin{mathpar}
  \inferrule[\RuleVerifyTrue]{%
    \emptystate,\ljM\,\ljV \reduce \cbState,\ljW\\
    \makeTruth{\ljW} = \ljTrue
  }{%
    \cbState,\fctx{\assert{\inV\termSVal}\cbVar{\defFlatC{\ljM}}}
    \translateB
    \cbState,\fctx{\assert{\inV\termSVal}\cbVar\top}
  }

  \inferrule[\RuleVerifyFalse]{%
    \emptystate,\ljM\,\ljV \reduce \cbState,\ljW\\
    \makeTruth{\ljW} = \ljFalse
  }{%
    \cbState,\fctx{\assert{\inV\termSVal}\cbVar{\defFlatC{\ljM}}}
    \translateB
    \cbState,\fctx{\assert{\inV\termSVal}\cbVar\bot}
  }
\end{mathpar}

\caption{Baseline transformation rules.}
\label{fig:baseline/transformation}
\end{figure*}

Figure~\ref{fig:baseline/transformation} specifies the baseline transformation
rules for $\Con$ terms $\ljM$. The first rules, \RefTirName{\RuleUnfoldAssert},
\RefTirName{\RuleUnfoldUnion}, and \RefTirName{\RuleUnfoldIntersection}, unfold
contracts and create constraints equivalent to the dynamics in
Subsection\ \ref{sec:lcon/semantics}.

Rule \RefTirName{\RuleUnfoldOp} unpacks a delayed contract on a term used in a
base type operation. As this contract is never checked, it is true by definition
(cf.\ \cite{KeilThiemann2015-blame}). Here, terms $\termTI$ and $\termTQ$ are
canonical terms with a possible immediate or delayed contract.

Rules \RefTirName{\RuleUnfoldDFunction} and
\RefTirName{\RuleUnfoldDIntersection} unfold a delayed contract that appears
left in an application corresponding to the dynamics. However, a function
contract unfolds the domain portion to the term on the right, which is not
necessarily a value.

Rule \RefTirName{\RuleUnroll} unrolls a delayed contract on a term $\termT$
right in an application. This step removes the contract from $\termT$ and
grafts it to all uses of $\termT$. Even though this steps duplicates the number
of contract assertions, we do not have more contract checks.

Rule \RefTirName{\RuleLower} create a new function contract by lowering the
contract on the function's body. The new function contract contains $\top$ on
it's domain portion as it accepts every argument. 

Rules \RefTirName{\RuleReverseI} and \RefTirName{\RuleReverseFalse} reverse the
order of a delayed contract nested in an assertion with an immediate contract or
$\bot$. This step is required to evaluate immediate contracts on values and to
unfold or unroll function's contracts.

Rule \RefTirName{\RuleReverseIf} pushes a contract out of a condition if the
contracts appears on both branches of the condition.

Rule \RefTirName{\RuleConvertTrue} reduces all occurrences of $\top$ and produces
a constraint that maps $\ljTrue$ to the corresponding blame variable $\cbVar$.
Knowledge about satisfied contracts can safely be removed from the source
program because they do not compromise the blame assignment.\footnote{%
  Due to the specification of contract satisfaction it is also possible to
  remove $\top$ from the source program without creating a constraint. A 
  non-existing constraint is handled as true because every non-violated
  contract assertion is satisfied, either because it is fulfilled or never
  verified.
}

In contrast to $\top$, $\bot$ (the knowledge about a failing contract) must
remain in the source program and cannot be translated into a constraint
statically. Doing this would over-approximate contract failures as the
constrain might lead to a blame state, without knowing if the ill-behaved
sources is ever executed. However, over-approximating violations might be
intended, as Appendix\ \ref{sec:appendix/success} demonstrates.

Finally, rules \RefTirName{\RuleVerifyTrue} and \RefTirName{\RuleVerifyTrue}
evaluate flat contracts on values that are available at compile time and
transform the outcome of the evaluation into $\top$ or $\bot$. Values might be
nested in a number of $\bot$ contracts of previously checked predicates.

\begin{lemma}[Progress]\label{thm:baseline/progress}
  For all terms $\ljM$ it holds that either $\ljM$ is a canonical term
  $\termT$ or there exists a transformation step
  $\cbState,\ljM\translateB\cbState',\ljN$.
\end{lemma}

\subsection{Strong Blame-preservation}
\label{sec:baseline/blame-preservation}

To sum up, the Baseline Transformation makes exactly the same transformation
steps at compile time, than the dynamics would later on do. Thus, $\translateB$
preserves the evaluation behavior of our program.

\begin{conjecture}[Strong Blame-preservation]\label{thm:strong-preservation}
  For all transformations $\cbState_M,\ljM \translateB \cbState_N,\ljN$ it hold
  that either $\cbState_M,\ljM \topreduce \cbState_M',\ljVal$ and
  $\cbState_N,\ljN \topreduce \cbState_N',\ljVal$ or that $\cbState_M,\ljM
  \topreduce \cbState_M'',\ljBlame^\cbBlame$ and $\cbState_N,\ljN \topreduce
  \cbState_N'',\ljBlame^\cbBlame$.
\end{conjecture}

%
%

\section{Subset Simplification}
\label{sec:subset}

The \emph{Subset Simplification} is the second tier in our transformation stack.
Its objective is to reuse knowledge of previously checked contracts to avoid
redundant checks whenever possible and to propagate contract violations to the
surrounding module boundary. 

The transformation first branches every alternative given by intersection and
union contracts into an individual observation on which every contract must be
fulfilled. Later it reduces contracts that are less restrictive and subsumed by
other contracts. Finally, it joins remaining fragments to a new contract where
possible.

To this end we introduce a subcontract relation which is closely related do
already existing definitions of naive subtyping. Subcontracting claims a
contract to be more restrictive than another contract.


\subsection{Predicates}
\label{sec:subset/predicates}

To achieve the best possible result it requires to have contracts and predicates
as fine grained as possible. Intersection and union contracts enable to build
complex contracts from different sub-contracts. Distinct properties can be
written in distinct predicates (flat contracts) and combined using intersection
and union.

For example, a contract that checks for positive even numbers can be written as
the intersection of one contract that checks for positive numbers and one
contract that checks for even numbers: 
$$
\defCapC{\defFlatC{\lambda\ljVar.(x >= 0)}}{\defFlatC{\lambda\ljVar.(x\%2 = 0)}}
$$
However, there is a twist. Fine-grained sub-contracts enable contracts be
optimized efficiently. But, at evaluation time, they must be handled in
separation such that we result in more predicate checks.

For example, consider the following contract that checks for natural numbers:
$$\defFlatC{\lambda\ljVar.(x >= 0)}$$ Obviously, one can write a similar
contract like this: $$\defCupC{\defFlatC{\lambda\ljVar.(x >
0)}}{\defFlatC{\lambda\ljVar.(x = 0)}}$$
The presence of intersection and union contracts enables developers to write
fine grained contracts, which, in turn, enable us to identify redundant parts
and to reduce already checked properties.

To avoid a run time impact cause by fine-grained predicates, we can easily
conjunct or disjunct predicates of the same blame label after optimization. This
is possible because the intersection and union of flat contracts corresponds to
the conjunction and disjunction of predicates (cf.
\cite{KeilThiemann2015-blame}).


\subsection{Contract Subsets}
\label{sec:subset/subset}

When optimizing contracts, we use subcontracting to characterize when a contract
is more restrictive than another contract, i.e.\ it subsumes all of its
obligations. The definition of subcontracting is closely related to naive
subtyping from the literature\ \cite{WadlerFindler2009}.

We write $\ConNSubseteq\conC\conD$ if $\conC$ is more restrictive than $\conD$,
i.e. $\ctx{\topassert\ljM\cbBlame\conC}\reduce^*\ljVal$ implies that
$\ctx{\topassert\ljM\cbBlame\conD}\reduce^*\ljVal$ and 
$\ctx{\topassert\ljM\cbBlame\conD}\reduce^*\blame\cbBlame$ implies that 
$\ctx{\topassert\ljM\cbBlame\conC}\reduce^*\blame\cbBlame$.

As a contract specifies the interface between a subject and its enclosing
context, subcontracting must be covariant for both parties: context \textbf{and}
subject.

For example, let \lstinline{Positive? = flat(,\x.x>0)} and 
\lstinline{Natural? = flat(,\x.x>=0)}. Than \lstinline{Positive?}
$\sqsubseteq^*$ \lstinline{Natural?} as for all \lstinline{x} it holds that
\lstinline{x>0} implies \lstinline{x>=0}.

For another example, \lstinline{Positive? -> Positive?} is a subcontract of
\lstinline{Natural? -> Natural?} as it is more restrictive on both portions.
Furthermore, \lstinline{Positive? -> Natural?} and
\lstinline{Natural? -> Positive?} are further subcontracts of
\lstinline{Positive? -> Positive?}.

In order to specify our subcontracting judgement, we factor subcontracting into
two subsidiary relations: subject subcontracting, written
$\ConSubseteqSubject\conC\conD$, and context subcontracting, written
$\ConSubseteqContext\conC\conD$. Relation $\ConSubseteqSubject\conC\conD$ indicates
that the subject portion in $\conC$ is more restrictive than the subject portion
in $\conD$, whereas $\ConSubseteqContext\conC\conD$ indicates that $\conC$
restricts the context more than $\conD$.

In contrast to Wadler and Findler's definition of positive and negative
subtyping (cf.\ \cite{WadlerFindler2009}), which is related to ordinary
subtyping, our definition of positive and negative subcontracting is reversal
because its is related to the naive subcontracting. We write
$\ConSubseteqContext\conC\conD$ if $\conC$ is more restrictive than $\conD$,
whereas Wadler and Findler write $\ConSubseteqContext\conD\conC$.

\begin{figure}
  \paragraph{Context Subcontract}
  \begin{mathpar}
    \inferrule[]{%
    }{%
      \ConSubseteqContext{\conC}{\conC}
    }

    \inferrule[]{%
    }{%
      \ConSubseteqContext{\conI}{\conJ}
    }

    \inferrule[]{%
      \ConSubseteqSubject{\conC_D}{\conD_D}\\
      \ConSubseteqContext{\conC_R}{\conD_R}
    }{%
      \ConSubseteqContext{\defFunC{\conC_D}{\conC_R}}{\defFunC{\conD_D}{\conD_R}}
    }

    \inferrule[]{%
      \ConSubseteqContext{\conC_D}{\conD_D}
    }{%
      \ConSubseteqContext{\defDepC{\defAbsC\ljVar\conC}}{\defDepC{\defAbsC\ljVar\conD}}
    }

    \inferrule[]{%
      \ConSubseteqContext{\conC_L}{\conD}\\ 
      \ConSubseteqContext{\conC_R}{\conD}
    }{%
      \ConSubseteqContext{\defCapC{\conC_L}{\conC_R}}{\conD}
    }

    \inferrule[]{%
      \ConSubseteqContext{\conC}{\conD_L} 
    }{%
      \ConSubseteqContext{\conC}{\defCapC{\conD_L}{\conD_R}}
    }

    \inferrule[]{%
      \ConSubseteqContext{\conC}{\conD_R}
    }{%
      \ConSubseteqContext{\conC}{\defCapC{\conD_L}{\conD_R}}
    }

    \inferrule[]{%
      \ConSubseteqContext{\conC_L}{\conD} 
    }{%
      \ConSubseteqContext{\defCupC{\conC_L}{\conC_R}}{\conD}
    }

    \inferrule[]{%
      \ConSubseteqContext{\conC_R}{\conD}
    }{%
      \ConSubseteqContext{\defCupC{\conC_L}{\conC_R}}{\conD}
    }

    \inferrule[]{%
      \ConSubseteqContext{\conC}{\conD_L}\\ 
      \ConSubseteqContext{\conC}{\conD_R}
    }{%
      \ConSubseteqContext{\conC}{\defCupC{\conD_L}{\conD_R}}
    }
  \end{mathpar}
  \vspace{\baselineskip}
  \paragraph{Subject Subcontract}
  \begin{mathpar}
    \inferrule[]{%
    }{%
      \ConSubseteqSubject{\conC}{\conC}
    }

    \inferrule[]{%
    }{%
      \ConSubseteqSubject{\conC}{\top}
    }

    \inferrule[]{%
    }{%
      \ConSubseteqSubject{\bot}{\conD}
    }

    \inferrule[]{%
      \ConSubsetTerm\ljM\ljN
    }{%
      \ConSubseteqSubject{\defFlatC{\ljM}}{\defFlatC{\ljN}}
    }

    \inferrule[]{%
      \ConSubseteqContext{\conC_D}{\conD_D}\\
      \ConSubsetSubject{\conC_R}{\conD_R}
    }{%
      \ConSubseteqSubject{\defFunC{\conC_D}{\conC_R}}{\defFunC{\conD_D}{\conD_R}}
    }

    \inferrule[]{%
      \ConSubseteqSubject{\conC_D}{\conD_D}
    }{%
      \ConSubseteqSubject{\defDepC{\defAbsC\ljVar\conC}}{\defDepC{\defAbsC\ljVar\conD}}
    }

    \inferrule[]{%
      \ConSubseteqSubject{\conC_L}{\conD}
    }{%
      \ConSubseteqSubject{\defCapC{\conC_L}{\conC_R}}{\conD}
    }

    \inferrule[]{%
      \ConSubseteqSubject{\conC_R}{\conD}
    }{%
      \ConSubseteqSubject{\defCapC{\conC_L}{\conC_R}}{\conD}
    }

    \inferrule[]{%
      \ConSubseteqSubject{\conC}{\conD_L}\\ 
      \ConSubseteqSubject{\conC}{\conD_R}
    }{%
      \ConSubseteqSubject{\conC}{\defCapC{\conD_L}{\conD_R}}
    }

    \inferrule[]{%
      \ConSubseteqSubject{\conC_L}{\conD}\\ 
      \ConSubseteqSubject{\conC_R}{\conD}
    }{%
      \ConSubseteqSubject{\defCupC{\conC_L}{\conC_R}}{\conD}
    }

    \inferrule[]{%
      \ConSubseteqSubject{\conC}{\conD_L}
    }{%
      \ConSubseteqSubject{\conC}{\defCupC{\conD_L}{\conD_R}}
    }

    \inferrule[]{%
      \ConSubseteqSubject{\conC}{\conD_R} 
    }{%
      \ConSubseteqSubject{\conC}{\defCupC{\conD_L}{\conD_R}}
    }

  \end{mathpar}

  \caption{Context and subject subcontracting.}
  \label{fig:subset/subcontract}
\end{figure}

Figure~\ref{fig:subset/subcontract} shows both relations. The two judgement
are defined in terms of each other and track swapping of responsibilities on
function arguments. Both judgements are reflexive and transitive.

We further write $\ConSubsetTerm\ljM\ljN$ if predicate $\ljM$ implies predicate
$\ljN$. For example, \lstinline{(x>0)} $\leq$ \lstinline{(x>=0)} and
\lstinline{(x=0)} $\leq$ \lstinline{(x>=0)}.

However, our subcontract judgement stops at $\ConSubsetTerm\ljM\ljN$. Relation
$\leq$ refers to an environment $\ConSubsetState$ that resolves relations on
terms. $\ConSubsetState$ might be given by a SAT solver or a set of predefined
relations.
\begin{mathpar}
    \inferrule[]{%
      (\ConSubsetTerm\ljM\ljN) \in \ConSubsetState
    }{%
      \ConSubsetTerm\ljM\ljN
    }
\end{mathpar}

It remains to defined subcontracting in terms of context and subject
subcontracting.
\begin{definition}\label{def:naive}
  $\conC$ is a \emph{naive subcontract} of $\conD$, written
  $\ConNSubseteq\conC\conD$, iff
  $$
  \ConSubseteqContext\conC\conD \quad\wedge\quad \ConSubseteqSubject\conC\conD.
  $$
\end{definition}

We further define an ordinary subcontracting judgement, $\ConSubseteq\conC\conD$,
which corresponds to the usuals definition of subtyping. Ordinary subcontracting
is contravariant for the context portion and covariant for the subject portion.
\begin{definition}\label{def:ordinary}
  $\conC$ is an \emph{ordinary subcontract} of $\conD$, written
  $\ConSubseteq\conC\conD$, iff
  $$
  \ConSubseteqContext\conD\conC \quad\wedge\quad \ConSubseteqSubject\conC\conD.
  $$
\end{definition}

In addition to the already mentiont notation, we sometimes use $\sqsubset$,
$\sqsubset^{*}$, $\sqsubset^{+}$, and $\sqsubset^{-}$ to indicate a proper
subcontract, analogous to already existing definitions on inequality.

\begin{lemma}\label{thm:preorder}
  Relation $\ConNSubseteq{}{}$ is a preorder, i.e. it is reflexive and transitive.
  For all $\conC$, $\conC'$, and $\conD$, we have that:
  \begin{itemize}
    \item $\ConNSubseteq\conC\conC$
    \item if $\ConNSubseteq\conC{\conC'}$ and $\ConNSubseteq{\conC'}\conD$ then
     $\ConNSubseteq\conC\conD$ 
  \end{itemize}
\end{lemma}

                                                                                                                          
\subsection{Canonical Terms, revisited}
\label{sec:subset/terms}

\begin{figure}[tp]
  \begin{displaymath}
    \begin{array}{lrl}

      \ctxT &\bbc& \ljHole \mid \forked\ctxT\ljM \mid \forked\termT\ctxT\\

      \ctxA &\bbc& \ljHole \mid \assert\ctxA\cbVar\conC\\
      \ctxB &\bbc& \ljHole \mid \ctxB\,\ljM \mid \termT\,\ctxB \mid
      \ljOp(\vec\termT\,\ctxB\,\vec\ljM) \mid \assert\ctxB\cbVar\conC
      \mid \ljIf\,\ctxB\,\ljM\,\ljN

    \end{array}
  \end{displaymath}

  \caption{%
    Parallel observations and contexts.
  }\label{fig:subset/syntax}
\end{figure}

As for the Baseline Transformation, we first define a set of canonical terms
(non-transformable) terms that serve as the output of this level's
transformation. The new canonical terms are a subset of the already defined
terms in Section~\ref{sec:baseline/terms}.
Appendix~\ref{sec:appendix/terms/subset} shows its definition.

In addition, Figure~\ref{fig:subset/syntax} shows the definition of a new
transformation contexts $\ctxT$ that looks through parallel observation. An
assertion context $\ctxA$ is a number of contract assertions and a body context
$\ctxB$ is a transformation context without lambda abstractions and without
conditions with a hole in a branch.

%
                                                            
\subsection{Subset Transformation}
\label{sec:subset/transformation}

\begin{figure*}
  \begin{displaymath}
    \begin{array}{llllllr}

      \DefTirName{\RuleForkUnion}
      & \cbState,
      & \fctx{\assert\termT\cbVar{(\defCupC\conC\conD)}}
      & \translateSInT
      & \append\cbState{(\defCb\cbVar\defCupC{\cbVar_1}{\cbVar_2})},
      & \fork{%
        \fctx{\assert\termT{\cbVar_1}\conC}
      }{%
        \fctx{\assert\termT{\cbVar_2}\conD}
      }
      & \cbVar_1,\cbVar_2\not\in{\cbState}
      \\

      \DefTirName{\RuleForkIntersection}
      & \cbState,
      & \fctx{(\assert{\termT_1}\cbVar{(\defCapC\conQ\conR)})\,\termT_2}
      & \translateSInT
      & \append\cbState{(\defCb\cbVar\defCapC{\cbVar_1}{\cbVar_2})},
      & \fork{%
        \fctx{(\assert{\termT_1}{\cbVar_1}\conQ)\,\termT_2)}
      }{%
        \fctx{(\assert{\termT_1}{\cbVar_2}\conR)\,\termT_2)}
      }
      & \cbVar_1,\cbVar_2\not\in{\cbState}
      \\
      \\

    \end{array}
  \end{displaymath}
  \begin{displaymath}
    \begin{array}{llllllr}

      \DefTirName{\RuleLift}
      & \cbState,
      & \fctx{\lambda\ljVar.\bctx{\assert\ljVar\cbVar\conI}}
      & \translateSInT
      & \append\cbState{(\defCb\cbVar{\neg\cbVar_1})},
      & \fctx{\assert{(\lambda\ljVar.\bctx{\ljVar})}{\cbVar_1}{(\defFunC\conI\top)}}
      & \cbVar_1\not\in{\cbState} 
      \\

      \DefTirName{\RuleLower}
      & \cbState,
      & \fctx{\lambda\ljVar.(\assert\termT\cbVar\conC)}
      & \translateSInT
      & \cbState
      & \fctx{\assert{(\lambda\ljVar.\termT)}\cbVar{(\defFunC\top\conC)}}
      & \termT\neq\ljVar\\
      \\

      \DefTirName{\RuleBlame}
      & \cbState,
      & \fctx{\lambda\ljVar.\bctx{\assert\termT\cbVar\bot}}
      & \translateSInT
      & \cbState,
      & \fctx{\lambda\ljVar.(\assert{\blame\cbBlame}\cbVar\bot})
      & \blame\cbBlame=\blameOf\cbVar\cbState\\

      \DefTirName{\RuleBlameIfTrue}
      & \cbState,
      & \fctx{\ljIf\,\termT_0\,\bctx{\assert{\termT_1}\cbVar\bot}\,\ljN}
      & \translateSInT
      & \cbState,
      & \fctx{\ljIf\,\termT_0\,(\assert{\blame\cbBlame}\cbVar\bot)\,\ljN}
      & \blame\cbBlame=\blameOf\cbVar\cbState\\

      \DefTirName{\RuleBlameIfFalse}
      & \cbState,
      & \fctx{\ljIf\,\termT_0\,\termT_1\,\bctx{\assert{\termT_2}\cbVar\bot}}
      & \translateSInT
      & \cbState,
      & \fctx{\ljIf\,\termT_0\,\termT_1\,(\assert{\blame\cbBlame}\cbVar\bot)}
      & \blame\cbBlame=\blameOf\cbVar\cbState\\
      
      \DefTirName{\RuleBlameGlobal}
      & \cbState,
      & \bctx{\assert\termT\cbVar\bot}
      & \translateSInT
      & \cbState,
      & \blame\cbBlame
      &\blame\cbBlame=\blameOf\cbVar\cbState\\
      \\

      \DefTirName{\RuleSubsetInner}
      & \cbState,
      & \fctx{\assert{\actx{\assert\termT{\cbVar_1}\conC}}{\cbVar_2}{\conD}}
      & \translateSInT
      & \cbState,
      & \fctx{\actx{\assert\termT{\cbVar_1}\conC}}
      & \conC\sqsubseteq^*\conD\\

      \DefTirName{\RuleSubsetOuter}
      & \cbState,
      & \fctx{\assert{\actx{\assert\termT{\cbVar_1}\conC}}{\cbVar_2}{\conD}}
      & \translateSInT
      & \cbState,
      & \fctx{\actx{\assert\termT{\cbVar_2}\conD}}
      & \conC\sqsupset^*\conD\\
      \\

      \DefTirName{\RuleReverseIf}
      & \cbState,
      & \fctx{(\ljIf\,\termT_0\,\inhole{\ctxA_1}{\assert{\termT_1}{\cbVar}{\conC}}\,\inhole{\ctxA_2}{\assert{\termT_2}{\cbVar}{\conC}})}
      & \translateSInT
      & \cbState
      & \fctx{\assert{(\ljIf\,\termT_0\,\inhole{\ctxA_1}{\termT_1}\,\inhole{\ctxA_2}{\termT_2})}{\cbVar}{\conC}}
      & 
      \\
      \\   

    \end{array}
  \end{displaymath}

  \begin{mathpar}

    \inferrule[\RuleMerge]{%
      \conQ=\defFunC\conC\top\\
      \conR=\defFunC\top\conD\\
      \blameOf{\cbVar_1}\cbState=\blameOf{\cbVar_2}\cbState\\
      \cbVar_3\not\in{\cbState}
    }{%
      \cbState,
      \fctx{\assert{\actx{\assert\termT{\cbVar_1}\conQ}}{\cbVar_2}{\conR}}\,
      \translateSInT\,
      \append{\append\cbState{\defCb{\cbVar_2}{\cbVar_3}}}{\defCb{\cbVar_1}{\cbVar_3}}\,
      \fctx{\actx{\assert\termT{\cbVar_3}\defFunC\conC\conD}}
    }

    \inferrule[\RuleBaseline]{%
      \cbState,\,
      \ljM\,
      \translateB\,
      \cbState',\,
      \ljN
    }{%
      \cbState,\,
      \ljM\,
      \translateSInT\,
      \cbState',\,
      \ljN
    }

    \inferrule[\RuleTrace]{%
      \cbState,\,
      \ljM\,
      \translateSInT\,
      \cbState',\,
      \ljN
    }{%
      \cbState,\,
      \inT\ljM\,
      \translateS\,
      \cbState',\,
      \inT\ljN
    }

  \end{mathpar}

  \caption{Subset transformation.}
  \label{fig:subset/transformation}
\end{figure*}

Figure~\ref{fig:subset/transformation} defines the subset transformation as an
extension of the baseline transformation in
Figure~\ref{fig:baseline/transformation}.

A significant change is that the new rules \RefTirName{\RuleForkUnion} and
\RefTirName{\RuleForkIntersection} replace the already exiting rules with the
name \RefTirName{\RuleUnfoldUnion} and \RefTirName{\RuleUnfoldDIntersection}.
Both rules split the current observation and handle both sides of an
intersection or union in separation. This step eliminates alternatives such that
in each branch every contract must be fulfilled.

Doing this enables us to remove a contracts whose properties are subsumed by
another contracts. Without splitting alternatives it would not be possible to
remove a contract on the basis of other contracts: The other contracts might be
nested in an alternative and thus it might never apply (as there is a weaker
alternative).

As for the baseline rules, unions are split immediately when asserted to a
value, whereas intersections must remain on the values until the values is used
in an application. This is because the context can choose between both
alternative on every use of the value.

Rule \RefTirName{\RuleUnfoldIntersection}
(Figure~\ref{fig:baseline/transformation}) must not be redefined, as the
intersections with flat contracts is equivalent to a conjunction of its
constituents (cf. \cite{KeilThiemann2015-blame}).

Rule \RefTirName{\RuleLift} lifts an immediate contract on a variable bound in a
lambda abstraction and produces a function contract on that abstraction.
However, at this level we are only allowed to lift contracts on variables
directly contained in the function body. Lifting from a deeper abstraction would
introduce checks that might never happen at run time.

Second, it is only allowed to lift an immediate contract. This is because a
lifted delayed contract would remain on the argument value and thus it would
effect all uses of that value, which entirely changes the meaning of the
program.

Rule \RefTirName{\RuleLower} pushes a contract on a function body down and
creates a new function contract from that contract. But now we have the
restriction that the function's return is not it's argument.

Rule \RefTirName{\RuleBlame} transforms a function body to a blame term
$\blame\cbBlame$ in case a contract violation happens in that body. As we split
alternatives we know that every contract must be fulfilled, and thus violating
a contract immediately results in a contract violation for that branch. However,
we cannot transform the whole execution to $\blame\cbBlame$ as we do not know
if the body es ever executed (e.g. because its is nested in a condition or
the lambda term is never used in an application).

Here, $\blameOf\cbVar\cbState$ computes a blame term $\blame\cbBlame$ resulting
from a predicate violation ($\defCb{\cbVar}{\ljFalse}$) in $\cbVar$ (cf.\ Appendix\
\ref{sec:appendix/auxiliary/blame}).

As before, $\assert\ljHole\cbVar\bot$ must remain on the $\blame\cbBlame$ term
to remember the violated contract. However, as the whole body transforms to
$\blame\cbBlame$, $\bot$ may stay as a contract on a function body. In this
case, rule \RefTirName{\RuleLower} pushes the information of a failing contract
down to the enclosing body, which in turn is unrolled when used in an
application and rule \RefTirName{\RuleBlame} transforms the enclosing context to
$\blame\cbBlame$. Appendix~\ref{sec:appendix/blame} shows an example reduction.

The rules \RefTirName{\RuleBlameIfTrue} and \RefTirName{\RuleBlameIfFalse} do
the same if a blame appears in branch of a conditions and rule
\RefTirName{\RuleBlameGlobal} produces a global blame term if $\bot$ appears on
a top-level term.

Rules \RefTirName{\RuleSubsetInner} and \RefTirName{\RuleSubsetOuter}
removes a contract from a term that is less restrictive than another contract on
that term. As this step removes a contract it might change the order of arising
contract violations. The weaker contract might raise a contract violation before
the stronger contract is checked. But, as the stronger contract remains, we know
that we definitely result in a blame state until the stronger contract is
satisfier. And, by definition, this implies that the weaker contract is also
satisfied.

Rule \RefTirName{\RuleReverseIf} pushes a contract out of an condition if the
contract remains ob both branches of the condition.

Rule \RefTirName{\RuleMerge} merges remaining fragments of the same source
contract to a new function. However, this can only be done if both contracts
have the same responsibility, which is indicated by precondition
$\blameOf{\cbVar_1}\cbState=\blameOf{\cbVar_2}\cbState$.

Finally, the rule \RefTirName{\RuleBaseline} lifts baseline transformations to
subset transformations, and rule \RefTirName{\RuleTrace} considers the
transformation in a parallel observation.

\begin{lemma}[Progress]\label{thm:subset/progress}
  For all terms $\ljM$ it holds that either $\ljM$ is a canonical term
  $\termT$ or there exists a transformation step
  $\cbState,\ljM\translateS\cbState',\ljN$.
\end{lemma}

Apart from the subset transformation which uses native subcontracting to reduces
contracts subsumed by other contract, we also use ordinary subcontracting to
simplifiy contracts contained in an intersection or union contract.
Appendix~\ref{sec:appendix/simplify} demonstrates this.


\subsection{Join Traces}
\label{sec:subset/join}

The subset transformation splits intersection and union into separated
observation, i.e.\ its outcome is a set of parallel observations. Thus, after
transforming all observations to canonical terms, we need to join the remaining
fragments to a valid source program.

Joining terms must be one of the last step as it is not possible to apply further
transformations afterwards. Subsequent transformations would mix up contracts
from different alternatives and thus they may change the meaning of a contract.

\begin{figure}[tp]
  \begin{displaymath}
    \begin{array}{lrl}

      \ctxM, \ctxN &\bbc& \ljHole \mid \lambda\ljVar.\ctxM \mid
      \ctxM\,\ljN \mid \ljM\,\ctxN \mid \ljOp(\vec\ljM\,\ctxM\,\vec\ljN)\\
      &\mid& \ljIf\,\ctxM\,\ljM\,\ljN \mid \ljIf\,\ljM\,\ctxM\,\ljN \mid
      \ljIf\,\ljM\,\ljN\,\ctxM \mid \assert\ctxL\cbVar\conC\\

      \ctxL, &\bbc& \lambda\ljVar.\ctxM \mid
      \ctxM\,\ljN \mid \ljM\,\ctxN \mid \ljOp(\vec\ljM\,\ctxM\,\vec\ljN)\\
      &\mid& \ljIf\,\ctxM\,\ljM\,\ljN \mid \ljIf\,\ljM\,\ctxM\,\ljN \mid
      \ljIf\,\ljM\,\ljN\,\ctxM \mid \assert\ctxL\cbVar\conC\\

    \end{array}
  \end{displaymath}
  \caption{%
    Final terms and contexts.
  }\label{fig:subset/final}
\end{figure}

To this end, Figure~\ref{fig:subset/final} defines a context $\ctxM$, which is
defined as usual as a term with a hole. The only exception is that contexts
$\ctxM$ did not have holes in a contract assertion.

As our transformation did not change the basic syntax of a program (it only
propagates contracts), terms remain equivalent (excepting contract assertions)
until they transform to a blame term. It follows that the only difference
between parallel terms are contract assertions and blame terms.

We call terms that differ only in contract assertions and blame terms as
\emph{structurally equivalent}. We write $\isEquivTerm{\ljM}{\ljN}$ if term
$\ljM$ is structurally equivalent to term $\ljN$ and $\isEquivCtx{\ctxG}{\ctxH}$
for the structural equivalence of contexts.

\begin{definition}\label{def:term-equivalence}
  Two terms are structural equal, written $\isEquivTerm{\ljM}{\ljN}$, if they only
  differ in contract assertions and blame terms.
\end{definition}

\begin{definition}\label{def:term-equivalence}
  Two contexts are structural equal, written $\isEquivCtx{\ctxM}{\ctxN}$, if they
  only differ in contract assertions and blame terms.
\end{definition}

\begin{figure} 
  \paragraph{Term Equivalence}
  \begin{mathpar}

    \inferrule[]{%
    }{%
      \isEquivTerm{\ljConst}{\ljConst}
    }

    \inferrule[]{%
    }{%
      \isEquivTerm{\ljVar}{\ljVar}
    }

    \inferrule[]{%
      \isEquivTerm{\ljM}{\ljN}
    }{%
      \isEquivTerm{\lambda\ljVar.\ljM}{\lambda\ljVar.\ljN}
    }

    \inferrule[]{%
      \isEquivTerm{\ljM}{\ljN}
    }{%
      \isEquivTerm{\lambda\ljVar.\ljM}{\lambda\ljVar.\ljN}
    }

    \inferrule[]{%
      \isEquivTerm{\ljM_0}{\ljN_0}\\
      \isEquivTerm{\ljM_1}{\ljN_1}
    }{%
      \isEquivTerm{\ljM_0\,\ljM_1}{\ljN_0\,\ljN_1}
    }

    \inferrule[]{%
      \isEquivTerm{\vec\ljM}{\vec\ljN}
    }{%
      \isEquivTerm{\ljOp(\vec\ljM)}{\ljOp(\vec\ljN)}
    }

    \inferrule[]{%
      \isEquivTerm{\vec\ljM_0}{\vec\ljN_0}\\
      \isEquivTerm{\vec\ljM_1}{\vec\ljN_1}\\
      \isEquivTerm{\vec\ljM_2}{\vec\ljN_2}
    }{%
      \isEquivTerm{\ljIf\,\ljM_0\,\ljM_1\,\ljM_2}{\ljIf\,\ljN_0\,\ljN_1\,\ljN_2}
    }

    \inferrule[]{%
      \isEquivTerm{\ljM}{\ljN}
    }{%
      \isEquivTerm{\assert\ljM\cbVar\conC}{\ljN}
    }

    \inferrule[]{%
      \isEquivTerm{\ljM}{\ljN}
    }{%
      \isEquivTerm{\ljM}{\assert\ljN\cbVar\conC}
    }

    \inferrule[]{%
    }{%
      \isEquivTerm{\ljM}{\blame\cbBlame}
    }

    \inferrule[]{%
    }{%
      \isEquivTerm{\blame\cbBlame}{\ljN}
    }
  \end{mathpar}  
  \vspace{\baselineskip}
  \paragraph{Context Equivalence}
  \begin{mathpar}
    \inferrule[]{%
    }{%
      \isEquivCtx{\ljHole}{\ljHole}
    }\and
    \inferrule[]{%
      \isEquivCtx{\ctxM}{\ctxN}
    }{%
      \isEquivCtx{\assert\ctxM\cbVar\conC}{\ctxN}
    }\and
    \inferrule[]{%
      \isEquivCtx{\ctxM}{\ctxN}
    }{%
      \isEquivCtx{\ctxM}{\assert\ctxN\cbVar\conC}
    }\and
    \inferrule[]{%
      \isEquivCtx{\ctxM}{\ctxN}
    }{%
      \isEquivCtx{\lambda\ljVar.\ctxM}{\lambda\ljVar.\ctxN}
    }\and
    \inferrule[]{%
      \isEquivCtx{\ctxM}{\ctxN}\\
      \isEquivTerm{\ljM}{\ljN}
    }{%
      \isEquivCtx{\ctxM\,\ljM}{\ctxN\,\ljN}
    }\and
    \inferrule[]{%
      \isEquivCtx{\ctxM}{\ctxN}\\
      \isEquivTerm{\ljM}{\ljN}
    }{%
      \isEquivCtx{\ljM\,\ctxM}{\ljN\,\ctxM}
    }\and
    \inferrule[]{%
      \isEquivCtx{\ctxM}{\ctxN}\\
      \isEquivTerm{\vec{\ljM_0}}{\vec{\ljN_0}}\\
      \isEquivTerm{\vec{\ljM_i}}{\vec{\ljN_i}}
    }{%
      \isEquivCtx{\ljOp(\vec{\ljM_0}\,\ctxM\,\vec{\ljM_i})}{\ljOp(\vec{\ljN_0}\,\ctxN\,\vec{\ljN_i})}
    }\and
    \inferrule[]{%
      \isEquivCtx{\ctxM}{\ctxN}\\
      \isEquivTerm{\ljM_1}{\ljN_1}\\
      \isEquivTerm{\ljM_2}{\ljN_2}
    }{%
      \isEquivCtx{\ljIf\,\ctxM\,\ljM_1\,\ljM_2}{\ljIf\,\ctxN\,\ljN_1\,\ljN_2}
    }\and
    \inferrule[]{%
      \isEquivCtx{\ctxM}{\ctxN}\\
      \isEquivTerm{\ljM_0}{\ljN_0}\\
      \isEquivTerm{\ljM_2}{\ljN_2}
    }{%
      \isEquivCtx{\ljIf\,\ljM_0\,\ctxM\,\ljM_2}{\ljIf\,\ljN_1\,\ctxN\,\ljN_2}
    }\and
    \inferrule[]{%
      \isEquivCtx{\ctxM}{\ctxN}\\
      \isEquivTerm{\ljM_0}{\ljN_0}\\
      \isEquivTerm{\ljM_1}{\ljN_1}
    }{%
      \isEquivCtx{\ljIf\,\ljM_0\,\ljM_1\,\ctxM}{\ljIf\,\ljN_0\,\ljN_1\,\ctxN}
    }
  \end{mathpar}
  \caption{Term and context equivalence.}
  \label{fig:subset/equivalence}
\end{figure}

Figure~\ref{fig:subset/equivalence} shows the structural equivalence of terms
and contexts. The equivalence relations looks through contract assertions and
compares only $\J$ terms. 

\begin{lemma}\label{thm:equivalence_relation}
  Relation $\isEquivCtx{}{}$ on terms $\ljM$ and context $\ctxM$ is an
  equivalence relation.
\end{lemma}

\begin{figure}
  \begin{mathpar}
    \inferrule[\RuleJoin]{%
      \forked{\ljM}{\ljN}\,
      \joinInT\,
      \forked{\ljM'}{\ljN'}
    }{%
      \inT{\forked{\ljM}{\ljN}}\,
      \join\,
      \inT{\forked{\ljM'}{\ljN'}}
    }\and  
    \inferrule[\RuleJoinUnit]{%
    }{%
      \forked{\ljM}{\ljM}
      \joinInT
      \ljM
    }\and  
    \inferrule[\RuleJoinTermLeft]{%
      \isEquivCtx{\ctxM}{\ctxN}\\
    }{%
      {\forked{%
        \inM{\inhole{\ctxA_l}{\termS}}
      }{%
        \inN{\inhole{\ctxA_r}{\blame\cbBlame}}
      }}
      \joinInT
      {\forked{%
        \inM{\inhole{\ctxA_l}{\termS}}
      }{
        \inN{\inhole{\ctxA_r}{\termS}}
      }}
    }\and
    \inferrule[\RuleJoinTermRight]{%
      \isEquivCtx{\ctxM}{\ctxN}\\
    }{%
      {\forked{%
        \inM{\inhole{\ctxA_l}{\blame\cbBlame}}
      }{%
        \inN{\inhole{\ctxA_r}{\termS}}
      }}
      \joinInT
      {\forked{%
        \inM{\inhole{\ctxA_l}{\termS}}
      }{
        \inN{\inhole{\ctxA_r}{\termS}}
      }}
    }\and
    \inferrule[\RuleJoinContract]{%
      \isEquivCtx{\ctxM}{\ctxN}\\
      \ctxA_l\neq\ctxA_r\\
      \ctxA=\ctxJoin{\ctxA_r}{\ctxA_l}\\
    }{%
      {\forked{%
        \inM{\inhole{\ctxA_l}{\termS}}
      }{%
        \inN{\inhole{\ctxA_r}{\termS}}
      }}
      \joinInT
      {\forked{%
        \inM{\inhole{\ctxA}{\termS}}
      }{
        \inN{\inhole{\ctxA}{\termS}}
      }}
    }
  \end{mathpar}
  \caption{Join parallel observations.}
  \label{fig:subset/join}
\end{figure}

To join terms in parallel observations we only need to walk through all terms
and to synchronise contract assertions and blame terms.
Figure~\ref{fig:subset/join} presents the synchronisation of contracts and
terms, indicated by relation $\termT\join\ljM$ on terms.

First, rule\ \RefTirName{\RuleJoin} triggers the synchronization of two
canonical terms nested in a parallel observation and rule\
\RefTirName{\RuleJoinUnit} dissolves a parallel observation if both terms are
identical.

The rules \RefTirName{\RuleJoinTermLeft} and \RefTirName{\RuleJoinTermRight}
handle the case that one side results in a blame state. In this case the blame
term is omitted and the operation proceeds with other term. Remember, it is not
required to maintain the blame term because $\assert\ljHole\cbVar\bot$ still
remains in the source program.

Finally, rule \RefTirName{\RuleJoinContract} looks for different contract
assertions on a term $\termS$ and merges the contract assertions. Precondition
$\ctxG\equiv\ctxH$ requires that the enclosing contexts are structurally
equivalent. Only the contract assertions must be different.

Here, $\ctxJoin{\ctxA_r}{\ctxA_l}$ computes the union of both contexts
and synchronizing proceeds with the new context in place.

The easiest way to merge contexts is to place one context in the hole of
the other context. Even through this did not change the blame behaviour of a
program it might duplicate contract checks, and thus it violates our ground rule
not to introduce more checks.

Thus, merging assertion contexts only places contract assertions in the hole
that are not already contained in the context. Formally, we specify the join of
two assertion contexts $\ctxA$ and $\ctxA'$ by
$$
\ctxJoin{\ctxA}{\ctxA'}=\inhole{\ctxA}{\ctxMinus{\ctxA'}{\ctxA}}
$$
Here, $\ctxMinus{\ctxA'}{\ctxA}$ computes a new context with assertions from
$\ctxA'$ that are not already contained in $\ctxA$.
Appendix~\ref{sec:appendix/auxiliary/context} shows its computation.

\subsection{Condense Remaining Contracts}
\label{sec:subset/condense}

Finally, one last step remains to be done: After synchronizing contracts it
might happen that we have identical contracts on the same term. As they belong
to different alternatives we cannot remove one of them, but we can condense them
a single assertion.

\begin{figure}
  \begin{mathpar}

   \inferrule[\RuleCondense]{%
    }{%
      \cbState,\,
      \inM{\assertWith{\actx{\assertWith\ljM{\cbBlame_0}{\cbVar_0}\conC}}{\cbBlame_1}{\cbVar_1}{\conC}}\,
      \condense\,
      \append\cbState{\defCb{\cbVar_1}{\cbVar_0}},\,
      \inM{\actx{\assertWith\ljM{\cbBlame_0}{\cbVar_0}\conC}}\,
    }

  \end{mathpar}

  \caption{Condense transformation.}
  \label{fig:subset/condense}
\end{figure}

Figure~\ref{fig:subset/condense} shows the transformation. If a term has two
assertions of the same contract, Rule\ \RefTirName{\RuleCondense} removes one of
them and creates a new constraint that redirect the result from the other
assertion to the blame variable of the remove assertion.


\subsection{Bringing all this together}
\label{sec:subset/together}

To model transformation with parallel observations, we define a new
transformation $\optimize$ on terms $\ljM$. It first applies the standard
transformation $\translateS$ until it reaches a canonical term $\termT$.
Second, it applies $\join$ to join terms and $\condense$ to condense the
remaining contracts. 
\begin{mathpar}
  \inferrule[]{%
    \emptystate,\ljM \translateS^* \cbState,\termT\\
    \termT \join^* \ljL\\
    \ljL \condense^* \ljN
  }{%
    \emptystate,\ljM
    \optimize^*
    \cbState,\ljN
  }
\end{mathpar}

\subsection{Weak Blame-preservation}
\label{sec:subset/blame-preservation}

While reorganizing contracts, the Subset Transformation did no strictly preserve
the blame behaviour of a program. An ill-behaved program by lead to another
contract violation first, whereas well-behaved programs still result in the same
output.

\begin{conjecture}[Weak Blame-preservation]\label{thm:weak-preservation}
  For all $\cbState_M,\ljM \optimize^* \cbState_N,\ljN$ it hold that either
  $\cbState_M,\ljM \topreduce \cbState_M',\ljVal$ and $\cbState_N,\ljN
  \topreduce \cbState_N',\ljVal$ or $\cbState_M,\ljM \topreduce
  \cbState_M'',\ljBlame^\cbBlame_M$ and $\cbState_N,\ljN \topreduce
  \cbState_N'',\ljBlame^{\cbBlame'}_N$.
\end{conjecture}

                                                                            
\section{Technical Results}
\label{sec:technical-results}

In addition to weak and string blame preservation, our contract simplification
grantees not to introduce more predicate checks at run time. Even throut we
reorganize and duplicate contract assertions, the total number of predicate
checks at run time remains the same or decrease on every transformation step.

\begin{definition}[Size]
  The size $\vert\ljL\vert$ of a closed term $\ljL$ is the number of predicates checks
  during reduction $\topreduce$ of $\ljL$.
\end{definition}

\begin{theorem}[Improvement]\label{thm:optimization}
  For each term $\ljM$ with $\emptystate,\ljM \optimize
  \cbState,\ljN$ it holds that $\vert\ljN\vert\leq\vert\ljM\vert$.
\end{theorem}

Finally, to prove soundness of our static contract simplification, we need to
show that our transformation terminates. 

\begin{theorem}[Termination]\label{thm:termination}
  For each term $\ljM$, there exists a canonical term $\termT$ such that
  $\cbState_M,\ljM \optimize^* \cbState_T,\termT$.
\end{theorem}


\section{Practical Evaluation}
\label{sec:evaluation}

To give an insight into the run time improvements of our system we apply the
simplifications to different versions of the \lstinline{addOne} example from
Section~\ref{sec:overview}. Our testing procedure uses the \lstinline{addOne}
function to increase a counter on each iteration in a while loop.

Each example program addresses a certain property and contains a different
number of contracts. For example, \emph{Example2} corresponds the
\lstinline{addOne} function in Section~\ref{sec:overview/unroll} and
\emph{Example4} corresponds to the function in
Section~\ref{sec:overview/subset}.
Appendix~\ref{sec:appendix/practical-evaluation} shows the JavaScript
implementation of all benchmark programs.

To run the examples we use the \TJS\ \cite{KeilThiemann2015-treatjs} contract
system for JavaScript and the SpiderMonkey\footnote{%
  \url{https://developer.mozilla.org/en-US/docs/Mozilla/Projects/SpiderMonkey}
} JavaScript engine. Lacking an implementation (beyond the
PLT Redex model) we applied the simplification steps manually.

\begin{figure}[t]
  \centering
  \small
  \begin{tabular}{ l || r || r@{~~}r || r@{~~}r}
    \toprule
    
    \textbf{Benchmark}& 
    \multicolumn{1}{c ||}{\textbf{Normal}}&
    \multicolumn{2}{c ||}{\textbf{Baseline}}&
    \multicolumn{2}{c}{\textbf{Subset}}\\

    &
    \textit{time (sec)}&
    \multicolumn{2}{c ||}{\textit{time (sec)}}&
    \multicolumn{2}{c}{\textit{time (sec)}}\\

    \midrule

    \emph{Example1}&
    39.4&
    27.0&  (-31.27\,\%)&
    27.0&  (-31.37\,\%)\\

    \emph{Example2}&
    87.1&
    58.4&  (-33.00\,\%)&
    46.1&  (-47.12\,\%)\\

    \emph{Example3}&
    66.5&
    54.4&  (-18.17\,\%)&
    26.5&  (-60.11\,\%)\\

    \emph{Example4}&
    114.5&
    85.1&  (-25.63\,\%)&
    44.6&  (-61.01\,\%)\\

    \emph{Example5}&
    148.3&
    108.0& (-27.18\,\%)&
    60.0&  (-59.55\,\%)\\

    \emph{Example6}&
    295.7&
    200.0&  (-32.36\,\%)&
    118.6&  (-59.90\,\%)\\

    \bottomrule
  \end{tabular}
  \caption{%
    Results from running the \TJS\ contract system.
    Column \textbf{Normal} gives the baseline execution time of the unmodified
    program, whereas column \textbf{Baseline} and column \textbf{Subset} contain
    the execution time and the improvement (in percent) after applying the
    baseline or subset simplification, respectively.
  }
  \label{fig:evaluation}
\end{figure}

Figure~\ref{fig:evaluation} shows the execution time required for a loop with
\lstinline{100000} iterations. The numbers indicate that the \emph{Baseline
Simplification} improves the run time by approximately 28\%, whereas the
\emph{Subset Simplification} makes an improvement up to 62\%.
Appendix~\ref{sec:appendix/practical-evaluation} shows the full table of
results.

In addition, there is another aspects not already addressed by
\citet{KeilThiemann2015-treatjs}: Adding contracts may also prevent a program
from being optimized. The SpiderMonkey engine, for example, uses two different
optimizing compilers to speed up operation. When activating SpiderMonkey's
optimizing JIT, the run time of the normal program without contracts decreases
by 20x, whereas the run time of the program with contracts only decreases by 2x.
This lack of optimization opportunities for the JIT is one important
reason for the big run time deterioration they reported.


\section{Related Work}
\label{sec:related-work}

\paragraph*{Higher-Order Contracts}

Software contracts were introduced with Meyer's \emph{Design by
Contract}{\texttrademark} methodology \cite{Meyer1988} which stipulates the
specification of Hoare-like pre- and postconditions for all components of a
program and introduces the idea of monitoring these contracts while the program
executes.

\citet{FindlerFelleisen2002} extend contracts and contract
monitoring to higher-order functional languages. Their work has attracted a
plethora of follow-up works that ranges from semantic investigations
\cite{BlumeMcallester2006,FindlerBlume2006} over studies on blame
assignment \cite{DimoulasFindlerFlanaganFelleisen2011,WadlerFindler2009} to
extensions in various directions: intersection and union contracts
\cite{KeilThiemann2015-blame}, polymorphic contracts \cite{AhmedFindlerSiekWadler2011,BeloGreenbergIgarashiPierce2011},
behavioral and temporal contracts
\cite{DimoulasTobin-HochstadtFelleisen2012,DisneyFlanaganMcCarthy2011},
etc.

\paragraph*{Contract Validation}

Contracts may be validated statically or dynamically. However, most frameworks
perform run-time monitoring as proposed in Meyer's work. 

Dynamic contract checking, as for example performed by
\citet{FindlerFelleisen2002,KeilThiemann2015-treatjs}, enable programmers to
specify properties of a component without restricting the flexibility
of the underlying dynamic programming language. However, run time
monitoring imposes a significant overhead as it extends the original program with
contract checks.

Purely static frameworks (e.g., ESC/Java
\cite{FlanaganLeinoLillibridgeNelsonSaxeStata2002}) transform specifications and
programs into verification conditions to be verified by a theorem prover. Others
\cite{XuPeytonJonesClaessen2009,Tobin-HochstadtHorn2012} rely on
program transformation and symbolic execution to prove adherence to contracts. 

Static contract checking avoids additional run-time checks, but
existing approaches are incomplete and  limited to check a restricted
set of properties. They rely on theorem proving engines to discharge
contracts that are written in the host language: the translation to
logic may not always succeed.

\paragraph*{Contract Simplification}

\citet{NguyenTobinHochstadtVanHorn2014} present a static contract
checker that has evolved from Tobin-Hochstadt and Van Horn's work on
symbolic execution \cite{TobinHochstadtVanHorn2012}. Their approach is to verify
contracts by executing programs on \emph{unknown} abstract values that are
refined by contracts as they are applied. If a function meets its obligation, the corresponding
contract gets removed. However, their approach only applies to the positive side
of a function contract.

Propagating contracts is closely related to symbolic executions. Both refine the
value of a term based on contracts, whether by moving the contract through the term
or by using an abstract value.

Compared to our work, \citet{NguyenTobinHochstadtVanHorn2014} addresses the opposite direction. Where we
unroll a contract to its enclosing context and decompose a contract into its
components, they verify the function's obligations based on the given domain
specification. However, in case the given contract cannot be verified, they must
retain the whole contract at its original position and they are not able to
simplify the domain portion of a function contract.

Furthermore, their symbolic verification is not able to handle true alternatives
in the style of intersection and union contracts.

We claim that both approaches are complementary to one another. Our
contract simplification would benefit from a preceding verification that
simplifies the function's obligations before unrolling a contract to the
enclosing context.

\citet{Xu2012} also combines  static and dynamic
contract checking. Her approach translates contracts into verification
conditions that get verified statically. Whereas satisfied conditions will be
removed, there may be conditions that cannot be proved: they remain in
the source program in the form of dynamic checks.

                                                  
\section{Conclusion}
\label{sec:conclusion}

The goal of static contract simplification is to reduce the overhead
of run-time contract checking.
To this end, we decompose contracts and propagates their components
through the program.
In many cases we are able to discharge parts of contracts
statically. 
The remaining components are simplified and reassembled to a new and
cheaper contract at a module boundary.

As a  novel aspect, we can simplify contracts that cannot be verified
entirely at compile time. The remaining contract fragments stay in the
program and get lifted to the enclosing 
module boundary.


A case study with microbenchmarks shows some promise. Our simplification 
decreases the total number of predicate checks at run time and thus reduces
the run-time impact caused by contracts and contract monitoring. The study
also shows that the degree of improvement depends on the granularity of
contracts. Fine-grained contracts enable more improvement.




\bibliographystyle{abbrvnat}



\cleardoublepage%

\appendix
%
                                          
\section{Constraint Satisfaction}
\label{sec:appendix/constraints}

The dynamics in Figure~\ref{fig:lcon/semantics} use constraints to create a
structure for computing positive and negative blame according to the semantics
of subject and context satisfaction, respectively. To this end, each blame
identifier $\cbIdent$ is associated with two truth values, $\subject{\cbIdent}$
and $\context{\cbIdent}$. Intuitively, if $\subject{\cbIdent}$ is false, then
the contract $\cbIdent$ is not subject-satisfied and may lead to positive blame
for $\cbIdent$. If $\context\cbIdent$ is false, then there is a context that
does not respect contract $\cbIdent$ and may lead to negative blame for
$\cbIdent$. 

An \emph{interpretation} $\cbSol$ of a constraint list $\cbState$ is a mapping
from blame identifiers to records of elements of $\cbTruth = \{\ljTrue,
\ljFalse\}$, such that all constraints are satisfied. We order truth values by
$\ljTrue \sqsubset \ljFalse$ and write $\sqsubseteq$ for the reflexive closure
of that ordering. This ordering reflects gathering of informations with each
execution step.

Formally, we specify the mapping by
$$\cbSol \in (\Rangeof{\cbIdent} \times \{\subjectName, \contextName\}) \to \cbTruth$$
where $\Rangeof{\cbIdent}$ ranges over metavariable $\cbIdent$ and constraint
satisfaction by a relation $\satisfies\cbSol\cbState$, which is specified in
Figure~\ref{fig:lcon/constraints}.

\begin{figure}[tp]
  \begin{displaymath}
    \makeTruth{\ljV} =
    \begin{cases}
      \ljFalse & \ljV = \ljFalse \\
      \makeTruth\ljW & \ljV = \assert\ljW\cbVar\conQ \\
      \ljTrue & \text{otherwise}
    \end{cases}
  \end{displaymath}
  \caption{Mapping values to truth values.}
  \label{fig:lcon/maketruth}
\end{figure}

\begin{figure}[tp]
  \begin{mathpar}
    \inferrule[\RuleCSEmpty]{}{%
      \satisfies\cbSol\emptystate
    }

    \inferrule[\RuleCSState]{%
      \satisfies\cbSol\cbCstr\\
      \satisfies\cbSol\cbState 
    }{%
      \satisfies\cbSol{\append\cbState\cbCstr}
    }

    \inferrule[\RuleCTIndirection]{%
      \cbSol(\subject\cbBlame) \sqsupseteq \cbSol(\subject\cbVar)\\
      \cbSol(\context\cbBlame) \sqsupseteq \cbSol(\context\cbVar)
    }{%
      \satisfies\cbSol{\defCb\cbBlame\cbVar}
    }

    \inferrule[\RuleCTFlat]{%
      \cbSol(\subject\cbIdent) \sqsupseteq \makeTruth\ljVal\\
      \cbSol(\context\cbIdent) \sqsupseteq \ljTrue
    }{%
      \satisfies\cbSol\defCb\cbIdent\ljVal
    }

    \inferrule[\RuleCTFunction]{%
      \cbSol(\subject\cbIdent) \sqsupseteq \cbSol(\context{\cbVar_1}
      \wedge (\subject{\cbVar_1} \Rightarrow \subject{\cbVar_2}))\\
      \cbSol(\context\cbIdent) \sqsupseteq \cbSol(\subject{\cbVar_1}
      \wedge \context{\cbVar_2})
    }{%
      \satisfies\cbSol{\defCb\cbIdent{\defFunC{\cbVar_1}{\cbVar_2}}}
    }

    \inferrule[\RuleCTIntersection]{%
      \cbSol(\subject\cbIdent) \sqsupseteq \cbSol(\subject{\cbVar_1}
      \wedge \subject{\cbVar_2}) \\
      \cbSol(\context\cbIdent) \sqsupseteq \cbSol(\context{\cbVar_1} 
      \vee \context{\cbVar_2})
    }{%
      \satisfies\cbSol{\defCb\cbIdent{\defCapC{\cbVar_1}{\cbVar_2}}}
    }

    \inferrule[\RuleCTUnion]{%
      \cbSol(\subject\cbIdent) \sqsupseteq \cbSol(\subject{\cbVar_1}
      \vee \subject{\cbVar_2})\\
      \cbSol(\context\cbIdent) \sqsupseteq \cbSol(\context{\cbVar_1}
      \wedge \context{\cbVar_2})
    }{%
      \satisfies\cbSol{\defCb\cbIdent{\defCupC{\cbVar_1}{\cbVar_2}}}
    }

    \inferrule[\RuleCTInversion]{%
      \cbSol(\subject\cbBlame) \sqsupseteq \cbSol(\context\cbVar)\\
      \cbSol(\context\cbBlame) \sqsupseteq \cbSol(\subject\cbVar)
    }{%
      \satisfies\cbSol{\defCb\cbBlame{\neg\cbVar}}
    }

  \end{mathpar}
  \caption{Constraint satisfaction.}
  \label{fig:lcon/constraints}
\end{figure}

In the premisses, the rules apply a constraint mapping $\cbSol$ to boolean
expressions over constraint variables. This application stands for the obvious
homomorphic extension of the mapping.

Every mapping satisfies the empty list of constraints (\RefTirName{\RuleCSEmpty}).
The concatenation of a constraint with a constraint list corresponds to the
intersection of sets of solutions (\RefTirName{\RuleCSState}). The indirection
constraint just forwards its referent (\RefTirName{\RuleCTIndirection}).

In rule \RefTirName{\RuleCTFlat}, $\ljW$ is the outcome of the predicate of a
flat contract. The rule sets subject satisfaction to $\ljFalse$ if
$\ljW=\ljFalse$ and otherwise to $\ljTrue$, where the function
$\makeTruth{\cdot} : \Rangeof\ljVal \to \cbTruth$ translates values to truth
values by stripping delayed contracts (see
Figure~\ref{fig:lcon/maketruth}). A flat contract never blames its
context so that $\context\cbIdent$ is always true.

The rule \RefTirName{\RuleCTFunction} determines the blame assignment for a
function contract $\cbIdent$ from the blame assignment for the argument and
result contracts, which are available through $\cbVar_1$ and $\cbVar_2$. Let's
first consider the subject part. A function satisfies contract $\cbIdent$ if it
satisfies its obligations towards its argument $\context{\cbVar_1}$ \textbf{and}
if the argument satisfies its contract then the result satisfies its contract,
too. The first part arises if the function is a higher-order function, which may
pass illegal arguments to its function-arguments. The second part is partial
correctness of the function with respect to its contract.

A function's context (caller) satisfies the contract if it passes an argument
that satisfies contract $\subject{\cbVar_1}$ \textbf{and} uses the result
according to its contract $\context{\cbVar_2}$.  The second part becomes
non-trivial with functions that return functions.

The rule \RefTirName{\RuleCTIntersection} determines the blame assignment for an
intersection contract at $\cbIdent$ from its constituents at
$\cbVar_1$ and $\cbVar_2$. A subject satisfies an intersection contract if it
satisfies both constituent contracts: $\subject{\cbVar_1} \wedge
\subject{\cbVar_2}$. A context, however, has the choice to fulfill one of the
constituent contracts: $\context{\cbVar_1} \vee \context{\cbVar_2}$.

Dually, the rule \RefTirName{\RuleCTUnion} determines the blame assignment for a
union contract at $\cbIdent$ from its constituents at $\cbVar_1$ and $\cbVar_2$.
A subject satisfies a union contract if it satisfies one of the constituent
contracts: $\subject{\cbVar_1} \vee \subject{\cbVar_2}$. A context, however,
needs to fulfill both constituent contracts: $\context{\cbVar_1} \wedge
\context{\cbVar_2}$, because it does not know which contract is satisfied by the
subject.

%

\clearpage

                                                                                                                          
\section{Canonical Terms}
\label{sec:appendix/terms}

Canonical (non-transformable) terms are the output of our transformation. Each
level has its own output grammar consisting of $\Con$ term $\ljM$ that cannot be
further transformed at this level.

Canonical terms distinguish contract-free term (indicated by latter $\termS$)
from terms with a contract (indicated by latter $\termT$). 


\subsection{Canonical Term of the Baseline Transformation}
\label{sec:appendix/terms/baseline}

\begin{figure}[tp]
  \paragraph{Canonical Terms}
  \begin{displaymath}
    \begin{array}{lrl}

      \termSVal &\bbc& \ljConst \mid \lambda\ljVar.\termS\\
      \termSNVal &\bbc& \ljVar \mid \blame\cbBlame\\
      &\mid&
      \termTI\,\termT \mid \termTConst\,\termT \mid \termTAbs\,\termTI \mid
      \termTAbs\,\termTVal\\
      &\mid& \ljOp(\vec\termTNQ)\\
      &\mid& \ljIf\,{\termT}_0\,{\termT}_1\,{\termT}_2
      \quad\textbf{where not}\quad      
      {\termT}_1=\assert{{\termT}_1'}\cbVar\conC \textbf{~and~}\\
      &&{\termT}_2=\assert{{\termT}_2'}\cbVar\conC\\
      \\

      \termS &\bbc& \termSVal \mid \termSNVal \\
      \\

      \termTConst &\bbc& \ljConst \mid \assert\termTConst\cbVar\bot\\
      \termTAbs &\bbc& \lambda\ljVar.\termS \mid \assert\termTAbs\cbVar\bot\\
      \termTVal &\bbc& \termSVal \mid \assert\termTVal\cbVar\bot\\
      \\

      \termTI &\bbc& \termSNVal \mid \assert{\termTI}\cbVar\conI \mid
      \assert{\termTI}\cbVar\bot\\
      \termTQ &\bbc& \termTVal \mid \termTI \mid \assert{\termTQ}\cbVar\conQ\\
      \termTNQ &\bbc& \termTVal \mid \termTI\\
      \\
      
      \termT &\bbc& \termTQ\\ 

    \end{array}
  \end{displaymath}
  \vspace{\baselineskip}
  \paragraph{Non-Canonical Terms}
  \begin{displaymath}
    \begin{array}{lrl}

      \termO, \termP  &\bbc& 
      \lambda\ljVar.\termO \mid \termO\,\ljN \mid \ljM\,\termO \mid 
      \ljOp(\vec\ljM\,\termO\,\vec\ljN)\\
      &\mid& \ljIf\,\termO\,\ljM\,\ljN \mid \ljIf\,\ljL\,\termO\,\ljN \mid
      \ljIf\,\ljL\,\ljM\,\termO \mid \assert\termO\cbVar\conC\\
      
      &\mid& (\inV{\lambda\ljVar.\ljM})\,(\assert\ljN\cbVar\conQ) \mid
      (\assert\ljM\cbVar\conQ)\,\ljN\\
      &\mid& \assert{\inV{\ljVal}}\cbVar\conI \mid
      \lambda\ljVar.(\assert\ljM\cbVar\conC)\\
      
      &\mid& \ljOp(\vec\ljL\,(\assert\ljM\cbVar\conQ)\,\vec\ljN) \mid
      \ljIf\,\ljL\,(\assert\ljM\cbVar\conC)\,(\assert\ljN\cbVar\conC)\\
      
      &\mid& \assert\ljM\cbVar\top \mid
      \topassert\ljM\cbBlame\conC \mid
      \assert\ljM\cbVar\defCupC\conC\conD \mid
      \assert\ljM\cbVar\defCapC\conI\conC\\
      
      &\mid& \assert{(\assert\ljM{\cbVar_q}\conQ)}{\cbVar_i}\conI \mid
      \assert{(\assert\ljM{\cbVar_q}\conQ)}{\cbVar_i}\perp

    \end{array}
  \end{displaymath}
  \caption{%
    Canonical and non-canonical terms of the \emph{Baseline Transformation}.
  }\label{fig:appendix/terms/baseline}
\end{figure}

This sections considers the output terms of the \emph{Baseline Transformation}
(Section~\ref{sec:baseline/transformation}).
Figure~\ref{fig:appendix/terms/baseline} defines the syntax of canonical terms
$\termS$ and $\termT$ as a subset of $\Con$ term $\ljM$.

A terms without a contract $\termS$ is either a value term $\termSVal$,
including first-order constants and lambda abstractions, or a non-value term
$\termSNVal$, which is either a variable, a blame term, an application, a
primitive operation, or a condition.

Applications are divided into different combinations. This excludes all
reducible combinations and only the given terms remain. For example, a function
contract on the left side of an application is reducible, whereas a function
contract in the right side is only reducible if the left side is a lambda term.

A canonical term $\termT$ is a term $\termTQ$, which is either a term
$\termTVal$ or a term  $\termTI$ with an indefinite number of delayed contracts
on the outermost position. Terms $\termTVal$ are all values with remaining
$\bot$ contracts. Terms $\termTI$ are all terms with immediate contracts or
$\bot$ contracts on the outermost position. As all immediate contracts on
values will be reduced, immediate contracts only remain on non-value terms
$\termSNVal$.

In addition to the canonical terms, Figure~\ref{fig:appendix/terms/baseline}
shows the syntax of transformable (non-canonical) terms $\termO$, which is the
complement of the canonical terms $\termT$.

A non-canonical term $\termO$ is either a term that contains a non-canonical
term or one of the patterns addressed by our transformation rules in
Section~\ref{sec:baseline/transformation}.

\begin{lemma}[Terms]\label{thm:baseline/terms}
  For all terms $\ljM$ it holds that $\ljM$ is either a canonical term
  $\termT$ or a transformable (non-canonical) term $\termO$.
\end{lemma}


\subsection{Canonical Term of the Subset Transformation}
\label{sec:appendix/terms/subset}

\begin{figure}[tp]
  \paragraph{Canonical Terms}
  \begin{displaymath}
    \begin{array}{lrl}

      \termSVal &\bbc& \ljConst \mid \lambda\ljVar.\termS
      \quad\textbf{where}\quad \lambda\ljVar.\termS \neq
      \lambda\ljVar.\inB{\assert{x}\iota\conI}\\

      \termSNVal &\bbc& \ljVar \mid \blame\cbBlame\\
      &\mid&
      \termTI\,\termTQ \mid \ljConst\,\termTQ \mid \termSVal\,\termTI \mid
      \termSVal\,\termS\\
      &\mid& \ljOp(\vec\termTNQ)\\
      &\mid& \ljIf\,{\termTQ}\,\inA{\termTQ'}\,\inA{\termTQ''}
      \quad\textbf{where not}\quad      
      {\termTQ'}=\assert{{\termT}'}\cbVar\conC \textbf{~and~}\\
      &&{\termTQ''}=\assert{{\termT}''}\cbVar\conC\\
      \\

      \termTI &\bbc& \termSNVal \mid \assert{\termSNVal}\cbVar\conI\\
      &\mid& \assertWith{\inA{\assertWith{\termTI}{\cbBlame_0}{\cbVar_0}\conI}}{\cbBlame_1}{\cbVar_1}\conJ
      \quad\textbf{where}\quad \conI\not\sqsubseteq^*\conJ, \conJ\not\sqsubseteq^*\conI\\
      \\

      \termTQ &\bbc& \termSVal \mid \termTI \mid \assert{\termSVal}\cbVar\conQ \mid \assert{\termTI}\cbVar\conQ\\
      &\mid& \assertWith{\inA{\assertWith{\termTQ}{\cbBlame_0}{\cbVar_0}\conQ}}{\cbBlame_1}{\cbVar_1}\conR
      \quad\textbf{where}\quad \conQ\not\sqsubseteq^*\conR, \conQ\not\sqsubseteq^*\conR\\
      && \textbf{~and not~} \conQ\neq\defFunC\conC\top
      ,\conR\neq\defFunC\top\conD, \cbBlame_0=\cbBlame_1\\
      \\

      \termT &\bbc& \termTQ \mid \forked{\termT_r}{\termT_l} \mid \assert{\blame\cbBlame}\cbVar\bot\\ 

    \end{array}
  \end{displaymath}
  \paragraph{Non-Canonical Terms}
  \begin{displaymath}
    \begin{array}{lrl}

      \termO, \termP &\abc& 
      
      \assertWith{\inA{\assertWith{\ljM}{\cbBlame_0}{\cbVar_0}\conC}}{\cbBlame_1}{\cbVar_1}\conD
      \quad\textbf{where}\quad
      \conC\not\sqsubseteq^*\conD,\conD\not\sqsubseteq^*\conC\\

      &\mid& \assertWith{\inA{\assertWith{\ljM}{\cbBlame}{\cbVar_0}\conC}}{\cbBlame}{\cbVar_1}\conD
      \quad\textbf{where}\quad
      \conC\not\sqsubseteq\conD,\conD\not\sqsubseteq\conC\\

      &\mid& \lambda\ljVar.\inB{\assert{x}\iota\conI} \\

      &\mid& \ljIf\,\ljL\,(\assert\ljM\cbVar\conC)\,(\assert\ljN\cbVar\conC)\\

      &\mid& \assert\ljM\cbVar\bot 
       
    \end{array}
  \end{displaymath}

  \caption{%
    Canonical and non-canonical terms of the \emph{Subset Transformation}.
  }\label{fig:appendix/terms/subset}
\end{figure}

Figure~\ref{fig:appendix/terms/subset} shows the definition of canonical term
$\termT$ and source terms $\termS$ as a restriction of the already defined
canonical terms. For readability, we reuse the same symbols as in
Section~\ref{sec:appendix/baseline/terms}. The repeatedly defined terms replace
the already existing definition, whereas every other definition remains valid.

A lambda abstractions $\lambda\ljVar.\termS$ now did not contain immediate
contracts on its argument $\ljVar$ and terms only contain contracts that are not
subsumed by another contract on that term. Terms $\termT$ now contains a blame
terms with a remaining $\bot$ contract.

\begin{lemma}[Terms]\label{thm:subset/terms}
  For all terms $\ljM$ it holds that $\ljM$ is either a canonical term
  $\termT$ or a transformable (non-canonical) term $\termO$.
\end{lemma}


%
%

\clearpage

\section{Auxiliary Functions and Definitions}
\label{sec:appendix/auxiliary}

This section shows some auxiliary functions and definitions used in the paper.

\subsection{Function $\rootof\cbIdent\cbState$}
\label{sec:appendix/auxiliary/root}

\begin{figure}[tp]
  $$
  \ctxK \bbc \ljHole \mid \append\ctxK\cbCstr
  $$
  \caption{%
    Constraint context.
  }\label{sec:auxiliary/context}
\end{figure}

Each top-level contract assertion is rooted in one blame label $\cbBlame$, i.e.
for every blame identifier $\cbIdent$ there exists a path from $\cbIdent$ to a
blame label $\cbBlame$. 

To computer the root label $\cbBlame$ for an identifier $\cbIdent$ we need to
walk backwards through all constraints $\cbCstr$ on a path until we reach the
top level blame label $\cbBlame$.

To do so, Figure\ \ref{sec:auxiliary/context} introduces a constraint context
$\ctxK$, which is either a hole $\ljHole$ or constraint $\cbCstr$ extending a
constraint context $\ctxK$.

\begin{figure}[tp]
  $$
  \rootof\cbIdent\cbState
  =
  \begin{cases}
    \cbBlame, & \cbIdent=\cbBlame\\
    \rootof{\cbIdent'}\cbState, & 
    \cbIdent=\cbVar,\inK{\append{\cbState}{(\defCb{\cbIdent'}\cbVar)}}\\
    \rootof{\cbIdent'}\cbState, & 
    \cbIdent=\cbVar_0,\inK{\append{\cbState}{(\defCb{\cbIdent'}\defFunC{\cbVar_0}{\cbVar_1})}}\\
    \rootof{\cbIdent'}\cbState, & 
    \cbIdent=\cbVar_1,\inK{\append{\cbState}{(\defCb{\cbIdent'}\defFunC{\cbVar_0}{\cbVar_1})}}\\
    \rootof{\cbIdent'}\cbState, & 
    \cbIdent=\cbVar_0,\inK{\append{\cbState}{(\defCb{\cbIdent'}\defCapC{\cbVar_0}{\cbVar_1})}}\\
    \rootof{\cbIdent'}\cbState, & 
    \cbIdent=\cbVar_1,\inK{\append{\cbState}{(\defCb{\cbIdent'}\defCapC{\cbVar_0}{\cbVar_1}}}\\
    \rootof{\cbIdent'}\cbState, & 
    \cbIdent=\cbVar_0,\inK{\append{\cbState}{(\defCb{\cbIdent'}\defCupC{\cbVar_0}{\cbVar_1})}}\\
    \rootof{\cbIdent'}\cbState, & 
    \cbIdent=\cbVar_1,\inK{\append{\cbState}{(\defCb{\cbIdent'}\defCupC{\cbVar_0}{\cbVar_1})}}\\
    \rootof{\cbIdent'}\cbState, & 
    \cbIdent=\cbVar, \inK{\append{\cbState}{(\defCb{\cbIdent'}{\neg\cbVar})}}\\
    \cbIdent, & \textit{otherwise}
  \end{cases}
  $$
  \caption{%
    Function \textit{root-of}.
  }\label{sec:auxiliary/rootof}
\end{figure}

Figure\ \ref{sec:auxiliary/rootof} finally shows the computation of function
$\rootof\cbIdent\cbState$. 

\subsection{Sign of a Blame Identifier}
\label{sec:appendix/auxiliary/sign}

Following the definition in Section\ \ref{sec:contracts}, constraints $\cbCstr$
implement a tree rooted in a blame label $\cbBlame$. Within the tree every blame
identifier $\cbIdent$ can report positive or negative blame, depending on its
constraint when mapping one of its truth values ($\subject\cbVar$ or
$\context\cbVar$) to false.

Thus, we define the \emph{sign} of a blame identifier $\cbIdent$ as the
responsibility to satisfy the subject part in this identifier $\cbIdent$. 
We represent signs by blame term $\ljBlame$ that would result from a subject
blame in $\cbIdent$. 
\begin{definition}
  $\ljBlame$ is the sign of a blame identifier $\cbIdent$ in $\cbState$ if
  $\append\cbState{(\defCb{\cbIdent}{\ljFalse})}$ leads to $\blame\cbBlame$ on a
  single path for any $\cbBlame$.
\end{definition}
In this case, every alternative must be treated as neutral, i.e. they do not
stop or change the arising blame. Thus we consider only a single path from
$\cbIdent$  to its root $\cbBlame$.

\begin{figure}[tp]
  $$
  \inv\ljBlame
  =
  \begin{cases}
    \ljNBlame, & \ljBlame=\ljPBlame\\
    \ljPBlame, & \ljBlame=\ljNBlame\\
  \end{cases}
  $$
  \caption{%
    Blame inversion.
  }\label{sec:auxiliary/inversion}
\end{figure}

\begin{figure}[tp]
  $$
  \signof\cbIdent\cbState
  =
  \begin{cases}
    \ljPBlame, & \cbIdent=\cbBlame\\
    \signof{\cbIdent'}\cbState, & 
    \cbIdent=\cbVar,\inK{\append{\cbState}{(\defCb{\cbIdent'}\cbVar)}}\\
    \inv\signof{\cbIdent'}\cbState, & 
    \cbIdent=\cbVar_0,\inK{\append{\cbState}{(\defCb{\cbIdent'}\defFunC{\cbVar_0}{\cbVar_1})}}\\
    \signof{\cbIdent'}\cbState, & 
    \cbIdent=\cbVar_1,\inK{\append{\cbState}{(\defCb{\cbIdent'}\defFunC{\cbVar_0}{\cbVar_1})}}\\
    \signof{\cbIdent'}\cbState, & 
    \cbIdent=\cbVar_0,\inK{\append{\cbState}{(\defCb{\cbIdent'}\defCapC{\cbVar_0}{\cbVar_1})}}\\
    \signof{\cbIdent'}\cbState, & 
    \cbIdent=\cbVar_1,\inK{\append{\cbState}{(\defCb{\cbIdent'}\defCapC{\cbVar_0}{\cbVar_1}}}\\
    \signof{\cbIdent'}\cbState, & 
    \cbIdent=\cbVar_0,\inK{\append{\cbState}{(\defCb{\cbIdent'}\defCupC{\cbVar_0}{\cbVar_1})}}\\
    \signof{\cbIdent'}\cbState, & 
    \cbIdent=\cbVar_1,\inK{\append{\cbState}{(\defCb{\cbIdent'}\defCupC{\cbVar_0}{\cbVar_1})}}\\
    \inv\signof{\cbIdent'}\cbState, & 
    \cbIdent=\cbVar, \inK{\append{\cbState}{(\defCb{\cbIdent'}{\neg\cbVar})}}
  \end{cases}
  $$
  \caption{%
    Function \textit{sign-of}.
  }\label{sec:auxiliary/signof}
\end{figure}

To compute the sign of a blame identifier $\cbIdent$, Figure\
\ref{sec:auxiliary/inversion} first defines an inversion of blame terms
$\ljBlame$ that flips the responsibility. Figure\
\ref{sec:auxiliary/signof} than shows the computation of the sign of blame
identifier $\cbIdent$. It is defined by induction and tracks swapping of
responsibilities on function and inversion constraints.

\subsection{Computation of Blame Terms}
\label{sec:appendix/auxiliary/blame}

To computes a blame term $\blame\cbBlame$ resulting from a predicate violation
($\defCb{\cbVar}{\ljFalse}$) in $\cbVar$ we need to compute the sign of $\cbVar$
in $\cbState$ and to annotate the resulting blame term with the root blame
identifier $\cbBlame$ of $\cbVar$ in $\cbState$.
\begin{mathpar}
  \inferrule[]{%
    \ljBlame=\signof\cbIdent\cbState\\
    \cbBlame=\rootof\cbIdent\cbState
  }{%
    \blameOf\cbVar\cbState=\blame\cbBlame
  }
\end{mathpar}

\subsection{Compute joined Assertion Context}
\label{sec:appendix/auxiliary/context}

We define the join of two assertion contexts $\ctxA$ and $\ctxA'$ by
$$
\ctxJoin{\ctxA}{\ctxA'}=\inhole{\ctxA}{\ctxMinus{\ctxA'}{\ctxA}}
$$
Term $\ctxMinus{\ctxA'}{\ctxA}$ computes a new context with assertions from
$\ctxA'$ that are not already contained in $\ctxA$. Its computation is defined
as:
$$
\ctxMinus{\ctxA}{\ctxA'}=\begin{cases}
  \ctxMinus{\ctxA''}{\ctxA'}, & \ctxA=\assert{\ctxA''}\cbVar\conC \wedge
  \ctxIn{\assert\ljHole\cbVar\conC}{\ctxA'}\\
  \assert{(\ctxMinus{\ctxA''}{\ctxA'})}\cbVar\conC, & \ctxA=\assert{\ctxA''}\cbVar\conC \wedge
  \ctxNotIn{\assert\ljHole\cbVar\conC}{\ctxA'}\\
  \ctxA, & \textit{otherwise}
\end{cases}
$$
We write $\ctxIn{\assert\ljHole\cbVar\conC}{\ctxA}$ to indicate that the
assertion of $\conC$ under blame label $\cbBlame$ and blame variable $\cbVar$ is
contained in $\ctxA$, and $\ctxNotIn{\assert\ljHole\cbVar\conC}{\ctxA}$ for its
negation. Contract entailment is defined by:
\begin{mathpar}
  \inferrule[]{%
  }{%
    \ctxIn{\assert\ljHole\cbVar\conC}{\assert\ctxA\cbVar\conC}
  }\and
  \inferrule[]{%
    \ctxIn{\assert\ljHole\cbVar\conC}{\ctxA}
  }{%
    \ctxIn{\assert\ljHole\cbVar\conC}{\assert\ctxA\cbVar\conD}
  }
\end{mathpar}

\clearpage

\section{Picky Semantics}
\label{sec:appendix/picky}

The semantics in Section~\ref{sec:baseline} guarantee \emph{Strong
Blame-preservation} with respect to Lax evaluation semantics (cf.\
Section~\ref{sec:lcon/semantics}). Lax unpacks contracts from values that flow
to a predicate, i.e.\ a contract cannot violate another contract.

The \emph{Baseline Transformation} makes this intuitively when processing the
rule \RefTirName{\RuleReverseI}. It pushes a delayed contract out of an
assertion with an immediate contract, such that the predicate will not see the
delayed contract. This step is entirely required for following transformation
steps.

When considering Picky (or even Indy) evaluation semantics, a contract must
remain on a value when the value is used in a predicate. Consequently, we cannot
pushes a delayed contract out of an assertion with an immediate contract without
changing the blame behaviour of the underlying program.

\begin{figure}
  \begin{mathpar}

    \inferrule[\RuleReverseI]{%
    }{%
      \cbState,\,
      \fctx{\assert{(\assert\termT{\cbVar_1}\conQ)}{\cbVar_2}\conI}\,
      \translateB\,
      \cbState\,
      \fctx{\assert{(\assert\termT{\cbVar_2}(\extendI\conI{(\assert\ljHole{\cbVar_1}\conQ})))}{\cbVar_1}\conQ}
    }

  \end{mathpar}

  \caption{Changes for Picky Semantics.}
  \label{fig:appendix/picky/transformation}
\end{figure}

To enables Picky (or Indy) semantics, rule \RefTirName{\RuleReverseI}
(Figure~\ref{fig:baseline/transformation}) must be replayed by another rule that
unrolls a delayed contract in an immediate contract when pushing the delayed
contract. Figure~\ref{fig:appendix/picky/transformation} shows the new rule.

\begin{figure}
  \begin{mathpar}

    \inferrule[\RuleExtendPredicate]{%
      \ljM=\lambda\ljVar.\ljN
    }{%
      \cbState,\,\fctx{\extendI{\defFlatC{\ljM}}{(\assert\ljHole\cbVar\conQ})}
      \translateB
      \cbState,\,\fctx{\lambda\ljVar.\ljN[\ljVar\mapsto{(\assert\ljVar\cbVar\conQ)]}}
    }

    \inferrule[\RuleExtendNonPredicate]{%
      \ljM\neq\lambda\ljVar.\ljN
    }{%
      \cbState,\,\fctx{\extendI{\defFlatC{\ljM}}{(\assert\ljHole\cbVar\conQ})}
      \translateB
      \cbState,\,\fctx{\defFlatC{\ljM}}
    }

  \end{mathpar}

  \caption{Unroll delayed contract.}
  \label{fig:appendix/picky/extend}
\end{figure}

Here, term
$\assert\termT{\cbVar}(\extendI\conI{(\assert\ljHole{\cbVar'}\conQ}))$ is a new
intermediate term that unrolls a delayed contract in a predicate (cf.\ rule
\RefTirName{\RuleUnroll}). Figure~\ref{fig:appendix/picky/extend} shows its
semantics.

Rule\ \RefTirName{\RuleExtendPredicate} unrolls a delayed contract to all uses
of the value in a predicate and rule \RefTirName{\RuleExtendPredicate} proceeds
on the unchanged predicate in case that the predicate is not a lambda
abstraction.

Indeed, it would also be possible to replace predicate $\ljM$ by
$\assert\ljM\cbVar{\defFunC\conQ\top}$ as it enforces $\conQ$ when the predicate
is used in an application. However, unrolling the contract enables to apply further 
simplification steps in the function's body.

\clearpage

\section{Proofs of Theorems}
\label{sec:appendix/proofs}

\subsection{Proof of Lemma \ref{thm:baseline/progress}}

\begin{proof}[Proof of Lemma \ref{thm:baseline/progress}]
  Immediate from the definition of canonical and non-canonical terms
  (Section~\ref{sec:appendix/terms/baseline}).
\end{proof}

\subsection{Proof of Lemma \ref{thm:equivalence_relation}}

\begin{proof}[Proof of Lemma \ref{thm:equivalence_relation}]
  $\isEquivCtx{}{}$ is an equivalence relation if and only if it is reflexive,
  symmetric, and transitive. Proof by induction over $\isEquivCtx{}{}$.
\end{proof}

\subsection{Proof of Conjecture \ref{thm:strong-preservation}}


\newcommand{\renaming}{\alpha}
\newcommand{\inAlpha}[2]{\alpha[#1\mapsto#2]}

\newcommand{\cbeq}[2]{#1\cong#2}
\newcommand{\cbeqa}[2]{#1\cong_{\renaming}#2}

\newcommand{\cdExtension}{\preccurlyeq}
\newcommand{\cbext}[2]{#1\cdExtension#2}

\subsection{Congruence of States and Terms}

\begin{definition}\label{def:cbeq}
  Two configurations $(\cbState_M,\ljM)$ and $(\cbState_N,\ljN)$ are structural
  equivalent under renaming $\renaming$, written
  $\cbeqa{\cbState_M,\ljM}{\cbState_N,\ljN}$ with
  $\renaming=\emptyset\mid\renaming\subst\cbVar{\cbVar'}$ iff
  ${\cbState_M,\ljM}\cong_{\renaming'}{(\cbState_N,\ljN)\subst\cbVar{\cbVar'}}$
  with $\renaming=\renaming'[\cbVar\mapsto\cbVar']$.
\end{definition}

\begin{lemma}\label{thm:cong}
  For all configurations $(\cbState_M,\ljM)$ and $(\cbState_N,\ljN)$ with
  $\cbeq{\cbState_M,\ljM}{\cbState_N,\ljN}$ it holds that 
  $\cbState_M,\ljM \reduce \cbState_M',\ljM'$ and $\cbState_N,\ljN \reduce
  \cbState_N',\ljN'$  with $\cbeq{\cbState_M',\ljM'}{\cbState_N',\ljN'}$.
\end{lemma}

\begin{proof}[Proof of Theorem \ref{thm:cong}]
  Proof by induction over $\cbState$ and $\ljM$.
\end{proof}

\subsection{Congruence of States and Terms}

\begin{definition}\label{def:cbext}
  State $\cbState'$ is an effect-free extension of state $\cbState$, written
  $\cbext{\cbState}{\cbState'}$, iff $\cbState,\ljM \reduce \cbState'',\ljM'$
  and $\cbState',\ljM \reduce \cbState''',\ljM'$ with
  $\cbext{\cbState''}{\cbState'''}$. 
\end{definition}


\begin{lemma}\label{thm:effect-free}
  For terms $\ljM$, states $\cbState$, and transformations $\cbState,\ljM
  \optimize \cbState',\ljN$ it hold that $\cbext{\cbState}{\cbState'}$.
\end{lemma}

\begin{proof}[Proof of Theorem \ref{thm:cong}]
  Immediate by the definition of $\reduce$ and $\optimize$. Every new constrains
  is composed of fresh blame variable $\cbVar$.
  Proof by induction over $\optimize$.
\end{proof}

\begin{proof}[Proof Sketch for Conjecture \ref{thm:strong-preservation}]
  Let reduction output $O\bbc\ljVal\mid\ljBlame^\cbBlame$.
  For all terms $\ljM$ and states $\cbState_M$ with $\cbState_M,\ljM \translateB
  \cbState_N,\ljN$ it hold that $\cbState_M,\ljM \topreduce \cbState_M',O$ and
  $\cbState_N,\ljN \topreduce \cbState_N',O$ with
  $\cbeqa{\cbState_M'}{\cbState_N'}$.

  As our transformation only propagates contract assertions and did not modify
  the underlying program, we know that if $\cbState,\ljM\reduce^*\cbState',\ljV$ and
  $\cbState,\ljN\reduce^*\cbState'',\ljW$ then $\ljV=\ljW$.

  To prove preservation, we show that the transformed term results in the same
  reduction sequence than the original term and that the both, state and terms,
  are equivalent under variable renaming $I$.

  Proof by induction over $\translateB$.

  \begin{description}
    \item [Rule \RefTirName{\RuleUnfoldAssert}:]
      $\ljM=\topassert\termT\cbBlame\conC$, 
      $\ljN=\assert\termT\cbVar\conC$, and
      $\cbState_N=\append{\cbState_M}{(\defCb{\cbBlame}{\cbVar'})}$.
      \begin{enumerate}
        \item[(1)]
          From $\ljM$ and $\topreduce$ it follows that
          $\cbState_M,\topassert\termT\cbBlame\conC\reduce^*\cbState_M',\topassert\ljVal\cbBlame\conC$
          and by rule \RefTirName{\RuleAssert} it follows that
          $\cbState_M',\topassert\ljV\cbBlame\conC\reduce\append{\cbState_M'}{(\defCb{\cbBlame}{\cbVar})},\,\assert\ljV{\cbVar}\conC$.
        \item[(2)]
          From lemma~\ref{thm:cong} it follows the
          $\cbState_N,\assert\termT{\cbVar'}\conC\reduce^*\cbState_N',\assert\ljV{\cbVar'}\conC$
          with $\cbext{\cbState_M'}{\cbState_N'}$.
      \end{enumerate}
      Claim holds because
      $\cbeqa{\append{\cbState_M'}{(\defCb{\cbBlame}{\cbVar})},\assert\ljV{\cbVar}\conC}{\cbState_N',\assert\ljV{\cbVar'}\conC}$
      with $\renaming=\emptyset\subst{\cbVar}{\cbVar'}$.

    \item [Rule \RefTirName{\RuleUnfoldUnion}:]
      Equivalent to case  \RefTirName{\RuleUnfoldAssert}.

    \item [Rule \RefTirName{\RuleUnfoldIntersection}:]
      Equivalent to case  \RefTirName{\RuleUnfoldAssert}.

     \item [Rule \RefTirName{\RuleUnfoldDFunction}:]
       $\ljM=(\assert{\termT_1}\cbVar{(\defFunC\conC\conD)})\, \termT_2$, 
       $\ljN=(\assert{\termT_1\, (\assert{\termT_2}{\cbVar_1}\conC))}{\cbVar_2}\conD$, and
       $\cbState_N=\append{\cbState_M}{(\defCb\cbVar\defFunC{\cbVar_1'}{\cbVar_2'})}$.
       \begin{enumerate}
         \item[(1)]
           From $\ljM$ and $\topreduce$ it follows that
           $\cbState_M,(\assert{\termT_1}\cbVar{(\defFunC\conC\conD)}),\termT_2
           \reduce^*
           \cbState_M',(\assert{\ljV_1}\cbVar{(\defFunC\conC\conD)}),\termT_2$
           and
           $\cbState_M',(\assert{\ljV_1}\cbVar{(\defFunC\conC\conD)}),\termT_2
           \reduce^*
           \cbState_M'',(\assert{\ljV_1}\cbVar{(\defFunC\conC\conD)}),\ljV_2$.
           By rule \RefTirName{\RuleDFunction} it follows that
           $\cbState_M'',(\assert{\ljV_1}\cbVar{(\defFunC\conC\conD)}),\ljV_2
           \reduce
           \append{\cbState_M''}{(\defCb\cbVar\defFunC{\cbVar_1}{\cbVar_2})},
           (\assert{\ljV_1\,(\assert{\ljV_2}{\cbVar_1}\conC))}{\cbVar_2}\conD$.

         \item[(2)]
           From lemma~\ref{thm:cong}, $\ljN$ and $\topreduce$ it follows that
           $\cbState_N,(\assert{\termT_1\,(\assert{\termT_2}{\cbVar_1'}\conC))}{\cbVar_2'}\conD
           \reduce^*
           \cbState_N',(\assert{\ljV_1\,(\assert{\termT_2}{\cbVar_1'}\conC))}{\cbVar_2'}\conD$
           and 
           $\cbState_N',(\assert{\ljV_1\,(\assert{\termT_2}{\cbVar_1'}\conC))}{\cbVar_2'}\conD
           \reduce^*
           \cbState_N'',(\assert{\ljV_1\,(\assert{\ljV_2}{\cbVar_1'}\conC))}{\cbVar_2'}\conD$
           with $\cbeqa{\cbState_M''}{\cbState_N''}$.
       \end{enumerate}  
      Claim holds because
      $\cbeqa{\append{\cbState_M''}{(\defCb\cbVar\defFunC{\cbVar_1}{\cbVar_2})},
           (\assert{\ljV_1\,(\assert{\ljV_2}{\cbVar_1}\conC))}{\cbVar_2}\conD}
           {\cbState_N'',(\assert{\ljV_1\,(\assert{\ljV_2}{\cbVar_1'}\conC))}{\cbVar_2'}\conD}$
           with $\renaming=\emptyset\subst{\cbVar_1}{\cbVar_2'}\subst{\cbVar_2}{\cbVar_2'}$.

    \item [Rule \RefTirName{\RuleUnfoldDIntersection}:]
      Equivalent to case \RefTirName{\RuleUnfoldDFunction}.

    \item [Rule \RefTirName{\RuleUnroll}:]
      Immediate from rule \RefTirName{\RuleBeta}.

    \item [Rule \RefTirName{\RuleLower}:]
      Immediate from rule \RefTirName{\RuleDFunction}.

    \item [Rule \RefTirName{\RuleReverseI}:]
      Immediate from rule \RefTirName{\RuleFlat}.

    \item [Rule \RefTirName{\RuleReverseFalse}:]
      Immediate from rule \RefTirName{\RuleFlat}.

    \item [Rule \RefTirName{\RuleReverseIf}:]
      Immediate from rules \RefTirName{\RuleIfTrue} and \RefTirName{\RuleIfFalse}.

    \item [Rule \RefTirName{\RuleConvertTrue}:]
      Immediate from rule \RefTirName{\RuleUnit}.

  \end{description}

\end{proof}

\subsection{Proof of Lemma \ref{thm:preorder}}

\begin{proof}[Proof of Lemma \ref{thm:preorder}]
  $\ConNSubseteq{}{}$ is a preorder if and only if it is reflexive and transitive.
  Proof by induction over $\ConNSubseteq{}{}$.
\end{proof}

\subsection{Proof of Lemma \ref{thm:subset/progress}}

\begin{proof}[Proof of Lemma \ref{thm:subset/progress}]
  Immediate from the definition of canonical and non-canonical terms
  (Section~\ref{sec:appendix/terms/subset}).
\end{proof}

\subsection{Proof of Conjecture \ref{thm:weak-preservation}}

\begin{proof}[Proof Sketch for Conjecture \ref{thm:weak-preservation}]
  Equivalent to proof of Conjecture \ref{thm:strong-preservation}. However,
  certain rules change the order of contract checks. So wen need to consider two
  cases: one where the early contract is satisfied, and one where it is
  violated.
\end{proof}

\subsection{Proof of Theorem \ref{thm:optimization}}

\begin{proof}[Proof of Theorem \ref{thm:optimization}]
  Proof by induction over $\optimize$. It is immediate from the definition that
  all rules, expect rule \RefTirName{\RuleUnroll}, either decompose, eliminate,
  or graft a contract. It remains to shows that:
  \begin{enumerate}
    \item no rules duplicates contract checks, and
    \item no rule pushes contract checks from one function body (or branch in a
      condition) to the enclosing body.
  \end{enumerate}

  \begin{description}
    \item [Rule \RefTirName{\RuleUnroll}].
      This rule grafts the assertion of a delayed contract to all uses of the
      contracted value. From $\reduce$, we have that for all
      $\lambda\ljVar.\ljM\,(\assert\ljN\cbVar\conQ)$ first $\ljN\reduce^*\ljVal$
      and by rule \RefTirName{\RuleBeta} it follows that
      $\lambda\ljVar.\ljM\,(\assert\ljVal\cbVar\conQ)$ reduces to
      $\lambda\ljVar.\ljM\subst\ljVar{(\assert\ljVal\cbVar\conQ)}$.
      By rule \RefTirName{\RuleUnroll} we have that
      $\lambda\ljVar.\ljM\,(\assert\ljN\cbVar\conQ)$ translates to
      $\lambda\ljVar.\ljM\subst\ljVar{(\assert\ljVar\cbVar\conQ)}\,\ljN$.
      From $\reduce$, we now have that $\ljN\reduce^*\ljVal$ and by rule
      \RefTirName{\RuleBeta} if follows that
      $\lambda\ljVar.\ljM\subst\ljVar{(\assert\ljVar\cbVar\conQ)}\subst\ljVar\ljVal$,
      which is equivalent to the result of the original reduction.

    \item [Rule \RefTirName{\RuleLift}]. 
      This rule lifts an immediate contract nested in a function body to a new
      function contract. However, this contract never applies if the function is
      not executed. Thus, it preserves the number contract checks.

    \item [Rule \RefTirName{\RuleLower}]. 
      This rule lowers a contract to the function boundary. But, this contract
      never applies if the function is not executed. It also preserves the number
      contract checks.

    \item [Rule \RefTirName{\RuleReverseIf}]. 
      This rule pushes an immediate contract out of a condition, if the contract
      definitely applies to both branches. Thus, it preserves the number
      contract checks.

    \item [Otherwise] Every other rule either decomposes or eliminates a contract, or
      reverses the order of contract checks. Thus they reduces or preserves the
      total number contract checks at run time.

  \end{description}
\end{proof}

\begin{proof}[Proof of Theorem \ref{thm:termination}]
  The transformation terminates with a canonical term because
  \begin{inparaenum}[(1)]
  \item there are no cycles and 
  \item there is no transformation step that introduces new contract check.
  \end{inparaenum}

  \begin{enumerate}
    \item 
      All rules, except rule \RefTirName{\RuleUnroll}, decompose or eliminate a
      contract, or grafts a contract to an enclosing term. Only rule
      \RefTirName{\RuleUnroll} duplicates a contract to all uses of the
      contracted value. However, as the transformation did not lift delayed
      contracts, it is not a cycle.
    \item Immediate from Theorem \ref{thm:optimization}.     
  \end{enumerate}
\end{proof}

\clearpage

\section{Simplify Intersection and Union}
\label{sec:appendix/simplify}

Apart from the subset optimization, which reduces a contract assertion based on
more restrictive assertion, we can also use the subcontract relation to simplify
contracts before unfolding.

For example, the intersection contract
\lstinline{Positive? -> Natural? cap Natural? -> Positive?} can be simplified to
a single function contract \lstinline{Positive? -> Positive?}. This is because,
the context can choose to call the function with a natural or positive number.
However, in case the input is a natural number, but not a positive number, than
the range must be a positive number. And, in case the input is a positive
number, than it is also a natural number (as positive numbers are a proper
subset of natural numbers) and thus the output needs to satisfy both range
contract, i.e.\@ it must be a positive number.

So, we can simplify the intersection because
\lstinline{Natural? -> Positive?} is an ordinary subset of
\lstinline{Positive? -> Natural?} (cf. Section \ref{sec:subset/subset}). Its
domain is less or equal restrictive and its range is more or equal restrictive.

\begin{figure}
  \begin{mathpar}
    \inferrule[\RuleSimplifyUnion/1]{%
      \ConSubseteq\conC\conD
    }{%
      \cbState,\,
      \fctx{\topassert\ljVal\cbBlame{(\defCapC\conC\conD)}}\,
      \translateB\,
      \cbState,\,
      \fctx{\topassert\ljVal\cbBlame{\conD}}
    }

    \inferrule[\RuleSimplifyUnion/2]{%
      \ConSubseteq\conD\conC
    }{%
      \cbState,\,
      \fctx{\topassert\ljVal\cbBlame{(\defCapC\conC\conD)}}\,
      \translateB\,
      \cbState,\,
      \fctx{\topassert\ljVal\cbBlame{\conC}}
    }

    \inferrule[\RuleSimplifyIntersection/1]{%
      \ConSubseteq\conC\conD
    }{%
      \cbState,\,
      \fctx{\topassert\ljVal\cbBlame{(\defCapC\conC\conD)}}\,
      \translateB\,
      \cbState,\,
      \fctx{\topassert\ljVal\cbBlame{\conC}}
    }

    \inferrule[\RuleSimplifyIntersection/2]{%
      \ConSubseteq\conD\conC
    }{%
      \cbState,\,
      \fctx{\topassert\ljVal\cbBlame{(\defCapC\conC\conD)}}\,
      \translateB\,
      \cbState,\,
      \fctx{\topassert\ljVal\cbBlame{\conD}}
    }

  \end{mathpar}
  \caption{Contract simplification.}
  \label{fig:appendix/simplification}
\end{figure}

In general, we can simplify intersection and union contracts if one operand is
an ordinary subset of the other operand. Figure\
\ref{fig:appendix/simplification} shows some simplification rules that omits one
branch of an alternative if that alternative is subsumed by the other branch.

To keep the system deterministic, we must simplify contracts at assertion time
and thus we skip rule \RefTirName{\RuleUnfoldAssert} in this cases.

\clearpage

\section{Introducing new Success Contracts}
\label{sec:appendix/success}

Our core rules in Section~\ref{sec:baseline} and Section~\ref{sec:subset} follows
the overall guideline not to approximate contract violations. However,
developers are free to choose a higher degree of optimization, while
over-approximating contract failures at compile time and introducing new module
boundaries.

This level's rules create some kind of success contracts by over-approximating
contracts and by lifting contracts to the upper-most boundary. However, this is
not a real optimization as it might lift checks that never happen at run time.

\begin{figure}[tp]
  \begin{displaymath}
    \begin{array}{lrl}

      \termSVal &\bbc& \ljConst \mid \lambda\ljVar.\termS
      \quad\textbf{where}\quad \lambda\ljVar.\termS \neq
      \lambda\ljVar.\inB{\assert{x}\iota\conC}\\

      \termSNVal &\bbc& \ljVar \mid \blame\cbBlame\\
      &\mid&
      \termTI\,\termTQ \mid \ljConst\,\termTQ \mid \termSVal\,\termTI \mid
      \termSVal\,\termS\\
      &\mid& \ljOp(\vec\termTI) \mid \ljOp(\vec\termSVal)\\
      &\mid& \ljIf\,{\termTQ}\,\inA{\termTQ'}\,\inA{\termTQ''}
      \quad\textbf{where not}\quad      
      {\termTQ'}=\assert{{\termT}'}\cbVar\conC \textbf{~and~}\\
      &&{\termTQ''}=\assert{{\termT}''}\cbVar\conC\\
      \\

      \termTI &\bbc& \termSNVal \mid \assert{\termSNVal}\cbVar\conI\\
      &\mid& \assertWith{\inA{\assertWith{\termTI}{\cbBlame_0}{\cbVar_0}\conI}}{\cbBlame_1}{\cbVar_1}\conJ
      \quad\textbf{where}\quad \conI\not\sqsubseteq^*\conJ, \conJ\not\sqsubseteq^*\conI\\
      \\

      \termTQ &\bbc& \termSVal \mid \termTI \mid \assert{\termSVal}\cbVar\conQ \mid \assert{\termTI}\cbVar\conQ\\
      &\mid& \assertWith{\inA{\assertWith{\termTQ}{\cbBlame_0}{\cbVar_0}\conQ}}{\cbBlame_1}{\cbVar_1}\conR
      \quad\textbf{where}\quad \conQ\not\sqsubseteq^*\conR, \conQ\not\sqsubseteq^*\conR\\
      && \textbf{~and not~} \conQ\neq\defFunC\conC\top
      ,\conR\neq\defFunC\top\conD, \cbBlame_0=\cbBlame_1\\
      \\

      \termT &\bbc& \termTQ \mid \forked{\termT_r}{\termT_l}\\ 

    \end{array}
  \end{displaymath}
  \caption{%
    Canonical terms and contexts.
  }\label{fig:success/syntax}
\end{figure}

\begin{figure}

  \begin{mathpar}

    \inferrule[\RuleLift]{%
      \cbVar_1\not\in{\cbState}\\
      \ljVar\in\free{\gctx{\ljVar}}
    }{%
      \cbState,\,
      \fctx{\lambda\ljVar.\gctx{\assert\ljVar\cbVar\conC}}\,
      \translateAInT\,
      \append\cbState{(\defCb\cbVar{\neg\cbVar_1})},\,
      \fctx{\assert{(\lambda\ljVar.\gctx{\ljVar})}{\cbVar_1}\defFunC\conC\top}\,
    }

    \inferrule[\RuleBlame]{%
    }{%
      \cbState,\,
      \fctx{\assert\termT\cbVar\bot}\,
      \translateA\,
      \append\cbState{(\defCb{\cbVar}{\ljFalse})},\,
      \fctx{\termT}
    }

    \inferrule[\RuleSubset]{%
      \cbState,\,
      \ljM\,
      \translateSInT\,
      \cbState',\,
      \ljN
    }{%
      \cbState,\,
      \ljM\,
      \translateAInT\,
      \cbState',\,
      \ljN
    }

    \inferrule[\RuleTrace]{%
      \cbState,\,
      \ljM\,
      \translateAInT\,
      \cbState',\,
      \ljN
    }{%
      \cbState,\,
      \inT\ljM\,
      \translateA\,
      \cbState',\,
      \inT\ljN
    }

  \end{mathpar}

  \caption{Approximations rules.}
  \label{fig:success/transformation}
\end{figure}

Figure~\ref{fig:success/syntax} defines the syntax of our canonical terms as a
subset of $\Con$ term $\ljM$, and Figure~\ref{fig:success/transformation} shows
possible transformation.

Rule\ \RefTirName{\RuleLift} lifts every contract on an argument $\ljVar$. This
step over-approximates contract violations as it also lifts contracts nested in
conditions and function applications that might never apply.

Rule\ \RefTirName{\RuleBlame} immediately updates the constraint state with
information about a failing contract. This might result in a blame state, even
if the ill behaved term is never executed.

Finally, the rule \RefTirName{\RuleSubset} lifts to the \emph{Subset Reduction}
$\translateS$ and rule \RefTirName{\RuleTrace} unpacks parallel observation.
Obviously, rule \RefTirName{\RuleBlame} replaces the rule with the same name in
$\translateS$ and we are never allowed to unroll a once lifted contract.
Otherwise we would end up in a cycle.

\clearpage

\section{Propagating Blame}
\label{sec:appendix/blame}
\label{sec:overview/blame}

In the preceding example we have seen that all contracts of a failing
alternative can be removed and only one \lstinline{bot} must remain on the first
failing term. As we split alternatives in different observation where every
contract must be fulfilled, we can report a violation for that branch immediately
when observing a violation. Every other contract can be ignored as this branch 
definitely runs into an error.

However, this is not completely true. We are only allowed to remove contracts in
the same context as the violation occurs. We cannot treat a contract as violated
if the violation only occurs in a certain aspect of the program.

For example, consider the following source program.
\begin{lstlisting}[name=overview]
let f = 
((,\ addOne (,\ z (addOne z)))
 ((,\ plus (,\ z ((plus 0) z))) 
  [(,\ x (,\ y (+ x y)))
   @ (Positive? -> (Positive? -> Positive?))]))
\end{lstlisting} 
After applying some simplification steps, we obtain the following
intermediate term.
\begin{lstlisting}[name=overview]
let f = 
((,\ addOne (,\ z (addOne z)))
 ((,\ plus (,\ z ([(plus [0 @ bot]) @ (Positive? -> Positive?)] z))) 
  (,\ x (,\ y (+ x y)))))
\end{lstlisting} 
Even though we do not know that \lstinline{addOne} violates the contract on
\lstinline{plus}, we are not allowed treat the contract as violated because we
do not know if \lstinline{addOne} is ever executed. But, we can replace
\lstinline{addOne}'s body by a blame term that reports a violation whenever the
function is executed. In addition, a \lstinline{bot} remains on the functions
body.
\begin{lstlisting}[name=overview]
let f = 
((,\ addOne (,\ z (addOne z)))
 ((,\ plus (,\ z ([blame @ bot])))
  (,\ x (,\ y (+ x y)))))
\end{lstlisting}
Having a \lstinline{blame} term and \lstinline{bot} seems to be twofold, but
\lstinline{bot} can be lowerd to a new function contract on \lstinline{addOne},
which is than unrolled to all places where \lstinline{addOne} is used, as the
following example demonstrates.
\begin{lstlisting}[name=overview]
let f = 
((,\ addOne (,\ z ([addOne @ (top -> bot)] z)))
 ((,\ plus (,\ z blame)) (,\ x (,\ y (+ x y)))))
\end{lstlisting}
Now, the same transformation steps apply as for normal contract. The function
contract gets unrolled, whereby \lstinline{bot} gets asserted to another term.
If it so happens the next function body gets transformed to a blame term with a
remaining \lstinline{bot}.

Continuing in this way propagates the information of a failing contact through applications
and conditions until it reaches the outermost boundary. Within this boundary we
definitely run into a contract violation. The following example shows the
result.
\begin{lstlisting}[name=overview]
let f = 
[((,\ addOne (,\ z (addOne z)))
  ((,\ plus (,\ z blame)) (,\ x (,\ y (+ x y))))) @ (top -> bot)]
\end{lstlisting}

\clearpage

\section{An Evaluation of Contract Monitoring}
\label{sec:appendix/evaluation}

Dynamic contract checking impacts the execution time. Source of this impact is
\begin{inparaenum}[(1)]
\item that every contract extents the original source with additional checks
  that need to checked when the program executes and
\item that the contract monitor itself causes some runtime overhead.
\end{inparaenum}

To demonstrate, we consider runtime values obtained from the \TJS\ contract
framework for JavaScript \cite{KeilThiemann2015-treatjs}. To the best of our
knowledge, \TJS\ is the contract framework with the heaviest runtime
deterioration. Reason for this is its flexibility and its support for full
intersection and union contracts \cite{KeilThiemann2015-blame}. The presence of
intersection and union contracts require that a failing contract must not
signal a violation immediately. As a violation depends on combinations of
failures in different sub-contracts the contract monitor muss continue even if a
first error occurs and connect the outcome of each sub-contract with the
enclosing operations.

To evaluate the runtime deterioration, they applied there framework to benchmark
programs from the Google Octane 2.0 Benchmark
Suite\footnote{\url{https://developers.google.com/octane/}}. The benchmark
programs measure a JavaScript engine's performance by running a selection of
complex and demanding programs.

For testing, they inferred a function contract for every function expression
contained a program and assert the function contract when the program executes.

\begin{figure}[t]
  \centering
  \small
  \begin{tabular}{l || r | r}
    \toprule
    \textbf{Benchmark}&
    \textbf{Full}&
    \textbf{Baseline}\\
    \midrule
    Richards& 
    1 day, 18 hours, 21 min, 20 sec&
    8 sec\\
    DeltaBlue&
    2 days, 10 hours, 36 min, 49 sec&
    4 sec\\
    Crypto&
    9 sec&
    8 sec\\
    RayTrace&
    23 hours, 12 min, 37 sec&
    4 sec\\
    EarleyBoyer&
    1 min, 9 sec&
    53 sec\\
    RegExp&
    9 sec&
    8 sec\\
    Splay&
    19 min, 19 sec&
    3 sec\\
    SplayLatency&
    19 min, 19 sec&
    3 sec\\
    NavierStokes&
    11 sec&
    4 sec\\
    pdf.js&
    6 sec&
    6 sec\\
    Mandreel&
    5 sec&
    5 sec\\
    MandreelLatency&
    5 sec&
    5 sec\\
    Gameboy Emulator&
    4 hours, 28 min, 28 sec&
    5 sec\\
    Code loading &
    12 sec&
    9 sec\\
    Box2DWeb&
    5 hours, 19 min, 49 sec&
    6 sec\\
    zlib&
    11 sec&
    11 sec\\
    TypeScript&
    22 min, 46 sec&
    24 sec\\
  \end{tabular}
  \caption{%
    Timings from running the Google Octane 2.0 Benchmark Suite. Column
    \textbf{Full} shows the time required to complete with contract assertions.
    The last column \textbf{Baseline} gives the baseline time without contract
    assertion.}
  \label{fig:appendix/evaluation/timings}
\end{figure}

\begin{figure}[t]
  \centering
  \small
  \begin{tabular}{ l || r | r | r}
    \toprule
    \textbf{Benchmark}& 
    \textbf{Asssert}&
    \textbf{Internal}&
    \textbf{Predicate}\\ 
    \midrule
    Richards& 
    24&
    1599377224&
    935751200\\
    DeltaBlue&
    54&
    2319477672&
    1340451212\\
    Crypto&
    1&
    5&
    3\\
    RayTrace&
    42&
    687240082&
    509234422\\
    EarleyBoyer&
    3944&
    89022&
    68172\\
    RegExp&
    0&
    0&
    0\\
    Splay&
    10&
    11620663&
    7067593\\
    SplayLatency&
    10&
    11620663&
    7067593\\
    NavierStokes&
    51&
    48334&
    39109\\
    pdf.js&
    3&
    15&
    9\\
    Mandreel&
    7&
    57&
    28\\
    MandreelLatency&
    7&
    57&
    28\\
    Gameboy Emulator&
    3206&
    141669753&
    97487985\\
    Code loading &
    5600&
    34800&
    18400\\
    Box2DWeb&
    20075&
    172755100&
    11266494\\
    zlib&
    0&
    0&
    0\\
    TypeScript&
    4&
    12673644&
    8449090\\
    \bottomrule
  \end{tabular}
  \caption{%
    Statistic from running the Google Octane 2.0 Benchmark Suite. Column
    \textbf{Assert} shows the numbers of top-level contract assertions. Column
    \textbf{Internal} contains the numbers of internal contract assertions
    whereby column \textbf{Predicates} lists the number of predicate
    evaluations.}
  \label{fig:appendix/evaluation/statistics}
\end{figure}

Figure\ \ref{fig:appendix/evaluation/timings} contains the runtime values of all
benchmark programs in two configurations, which are explained in the figure's
caption. As expected, all run times increase when adding contracts.

The examples show that the run-time impact of contract monitor depends on the
program and on the particular value that is monitored. While some programs like
\emph{Richards}, \emph{DeltaBlue}, \emph{RayTrace}, and \emph{Splay} are heavily
affected, others are almost unaffected: \emph{Crypto}, \emph{NavierStokes}, and
\emph{Mandreel}, for instance.

It follows that the impact of a contract depends on the frequency of
its application. A contract on a heavily used function (e.g., in
\emph{Richards}, \emph{DeltaBlue}, or \emph{Splay}) causes a significantly
higher deterioration of the run-time. 

For better understanding, Figure\ \ref{fig:appendix/evaluation/statistics} lists
some numbers of internal counters. The numbers indicate that the heavily
affected benchmarks (\emph{Richards}, \emph{DeltaBlue}, \emph{RayTrace},
\emph{Splay}) contain a very large number of predicate checks. Other
benchmarks are either not affected (\emph{RegExp}, \emph{zlib}) or only slightly
affected (\emph{Crypto}, \emph{pdf.js}, \emph{Mandreel}) by contracts.

For example, the \emph{Richards} benchmark performs 24 top-level contract
assertions (this are all unique contracts in a source program), 1.6 billion
internal contract assertions (including top-level assertions, \emph{delayed}
contract checking, and predicate evaluation), and 936 million predicate
evaluation.

To sum up, we see that number of predicate checks is substantially responsible
for the runtime deterioration of a contract system. Thus, reducing the number of
predicate checks will entirely improve the runtime.

\clearpage

\section{Practical Evaluation, continued}
\label{sec:appendix/practical-evaluation} 

To demonstrate the run time improvement we applied our contract simplification
to a number simple examples, as introduces in Section\ \ref{sec:overview} and
Section\ \ref{sec:evaluation}. The testing procedure is a simple while loop that
uses different versions of an \lstinline{addOne} function to increase a counter
on every iteration.

This section shows the examples programs written in JavaScript and the full
table of results. To run the examples we use the \TJS\footnote{
  \url{https://github.com/keil/TreatJS}
} \cite{KeilThiemann2015-treatjs} contract system for JavaScript and the Mozilla
SpiderMonkey\footnote{
  \url{https://developer.mozilla.org/en-US/docs/Mozilla/Projects/SpiderMonkey}}
  engine. While there is actually no implementation of our transformation system
  for JavaScript, we applied all transformation manually and produces the
  simplified contracts by hand.

In \TJS{}, function \lstinline[language=JavaScript]{assert} asserts a contrat
(given as second argument) to a subject value (given as first argument).
Constructor \lstinline[language=JavaScript]{AFunction} creates a function
contract from a set of argument contracts (given as first argument) and a return
contract (given as second argument). Constructor
\lstinline[language=JavaScript]{Intersection} creates an intersection contract
from two other contracts.
\lstinline[language=JavaScript]{_Number},
\lstinline[language=JavaScript]{_Natural},
\lstinline[language=JavaScript]{_Positive},
\lstinline[language=JavaScript]{_Negative},
\lstinline[language=JavaScript]{_String} are flat contrasts checking for number
values, natural numbers, positive numbers, negative numbers, and string values,
respectively.

\subsection{The example programs}
\label{sec:appendix/practical-evaluation/examples}

\begin{example}[\syntax{addOne1}]\label{ex:addOne1}
  In a first example we add a simple function contract to \lstinline{plus}
  restricting its domain and range to number values. Every use of
  \lstinline{addOne1} triggers three predicates checks.

\begin{lstlisting}[language=JavaScript,name=evaluation]
var addOne1 = (function () {
  var plus = assert(function (x, y) {
    return x + y;
  }, AFunction([_Number, _Number], _Number))
  var addOne = function (x) {
    return plus (x, 1);
  }
  return addOne;
})();
\end{lstlisting}
\end{example}

\begin{example}[\syntax{addOne2}]\label{ex:addOne2}
  Our second example adds an intersection contract to \lstinline{plus}. As the
  native \lstinline{+} operator is overload and works for strings and numbers,
  our contract restricts the domain either to string or number values and
  ensures that the function has to return a string or a number value
  corresponding to the input. Every use of \lstinline{addOne2} triggers six
  predicate checks.

\begin{lstlisting}[language=JavaScript,name=evaluation]
var addOne2 = (function () {
  var plus = assert(function (x, y) {
    return x + y;
  }, Intersection(
      AFunction([_Number, _Number], _Number),
      AFunction([_String, _String], _String)
  ));
  var addOne = function (x) {
    return plus (x, 1);
  }
  return addOne;
})();
\end{lstlisting}
\end{example}

\begin{example}[\syntax{addOne3}]\label{ex:addOne3}
  The third example extends \lstinline{addOne1} by adding another function
  contract to \lstinline{addOne}. The contract restricts \lstinline{addOne}'s
  domain to natural numbers and requires a positive number as return. In
  combination, every use of \lstinline{addOne3} triggers 5 predicate checks.

\begin{lstlisting}[language=JavaScript,name=evaluation]
var addOne3 = (function () {
  var plus = assert(function (x, y) {
    return x + y;
  }, AFunction([_Number, _Number], _Number));
  var addOne =  assert(function (x) {
    return plus (x, 1);
  }, AFunction([_Natural], _Positive));
  return addOne;
})();
\end{lstlisting}
\end{example}

\begin{example}[\syntax{addOne4}]\label{ex:addOne4}
  The next example merges the intersection contract on \lstinline{plus} (from
  Example~\ref{ex:addOne2}) and the contract on \lstinline{addOne}
  (Example~\ref{ex:addOne3}). Every call of \lstinline{addOne4} leads to a total
  number 8 predicate checks.

\begin{lstlisting}[language=JavaScript,name=evaluation]
var addOne4 = (function () {
  var plus = assert(function (x, y) {
    return x + y;
  }, Intersection(
      AFunction([_Number, _Number], _Number),
      AFunction([_String, _String], _String)
  ));
  var addOne = assert(function (x) {
    return plus (x, 1);
  }, AFunction([_Natural], _Positive));
  return addOne;
})();
\end{lstlisting}
\end{example}

\begin{example}[\syntax{addOne5}]\label{ex:addOne5}
  In this example we also overload \lstinline{addOne} by making either string
  concatenation or addition, depending on \lstinline{addOne}'s input. While
  adding an intersection contract to \lstinline{addOne}, every use of
  \lstinline{addOne5} leads to 10 predicate checks.

\begin{lstlisting}[language=JavaScript,name=evaluation]
var addOne5 = (function () {
  var plus = assert(function (x, y) {
    return x + y;
  }, Intersection(
      AFunction([_Number, _Number], _Number),
      AFunction([_String, _String], _String)
  ));
  var addOne = assert(function (x) {
    return (typeof x == 'string') ? plus (x, '1') : plus (x, 1);
  }, Intersection(
       AFunction([_Natural], _Positive),
       AFunction([_String], _String)
  );
  return addOne;
})();
\end{lstlisting}
\end{example}

\begin{example}[\syntax{addOne6}]\label{ex:addOne6}
  Our last example simulates the case that we add fine-grained properties to
  \lstinline{plus} and \lstinline{addOne}. In this case we state different
  properties in different function contracts and use intersections to combine
  those properties. Before simplifying the contract, every call of
  \lstinline{addOne6} leads to a total number 17 predicate checks.

\begin{lstlisting}[language=JavaScript,name=evaluation]
var addOne6 = (function () {
  var plus = assert(function (x, y) {
    return x + y;
  }, Intersection(
      Intersection(
        AFunction([_Number, _Number], _Number),
        AFunction([_String, _String], _String)),
      Intersection(
        Intersection(
         AFunction([_Natural, _Positive], _Positive),
         AFunction([_Positive, _Natural], _Positive)),
        AFunction([_Negative, _Negative], _Negative))));
  var addOne = assert(function (x) {
    return plus (x, 1);
  }, AFunction([_Natural], _Positive));
  return addOne;
})();
\end{lstlisting}
\end{example}

\subsection{The Run-Time Improvements}
\label{sec:appendix/practical-evaluation/improvements}

\begin{figure*}[t]
  \centering
  \small
  \begin{tabular}{ l || r | r || r r | r r || r r | r r}
    \toprule
       
    \multicolumn{11}{c}{\textbf{\textit{Full}}}\\

    \midrule
 
    \textbf{Benchmark}& 
    \multicolumn{2}{c ||}{\textbf{Normal}}&
    \multicolumn{4}{c ||}{\textbf{Baseline}}&
    \multicolumn{4}{c}{\textbf{Subset}}\\

    &
    \textit{time (ms)}&
    \textit{predicates}&
    \multicolumn{2}{c |}{\textit{time (ms)}}&
    \multicolumn{2}{c ||}{\textit{predicates}}&
    \multicolumn{2}{c |}{\textit{time (ms)}}&
    \multicolumn{2}{c}{\textit{predicates}}\\

    \midrule

    Without Contracts&
    2& 0\\

    Proxy only&
    142.33& 0\\

    \midrule

    Example~\ref{ex:addOne1}&
    39404&	300000&
    27081&  (- 31.27\%)&
    200000& (- 33.33\%)&
    27043&  (- 31.37\%)&
    200000& (- 33.33\%)\\

    Example~\ref{ex:addOne3}&
    87143&  600000&	
    58385&  (- 33.00\%)&
    400000& (- 33.33\%)&
    46085&  (- 47.12\%)&
    300000& (- 50.00\%)\\

    Example~\ref{ex:addOne4}&
    66474&	500000&
    54396&  (- 18.17\%)&
    400000& (- 20.00\%)&
    26518&  (- 60.11\%)&
    200000& (- 60.00\%)\\ 

    Example~\ref{ex:addOne4}&
    114468& 800000&
    85126&  (- 25.63\%)&
    600000& (- 25.00\%)&
    44633&  (- 61.01\%)&
    300000& (- 62.50\%)\\

    Example~\ref{ex:addOne5}&
    148249& 1000000&
    107956& (- 27.18\%)&
    800000& (- 20.00\%)&
    59970&  (- 59.55\%)&
    500000& (- 50.00\%)\\

    Example~\ref{ex:addOne6}&
    295682& 1700000&
    200009&  (- 32.36\%)&
    1200000& (- 29.41\%)&
    118579&  (- 59.90\%)&
    700000&  (- 58.82\%)\\

    \bottomrule
  \end{tabular}

  \begin{tabular}{ l || r | r || r r | r r || r r | r r}
    \toprule

    \multicolumn{11}{c}{\textbf{\textit{No-Ion}}}\\

    \midrule

    \textbf{Benchmark}& 
    \multicolumn{2}{c ||}{\textbf{Normal}}&
    \multicolumn{4}{c ||}{\textbf{Baseline}}&
    \multicolumn{4}{c}{\textbf{Subset}}\\

    &
    \textit{time (ms)}&
    \textit{predicates}&
    \multicolumn{2}{c |}{\textit{time (ms)}}&
    \multicolumn{2}{c ||}{\textit{predicates}}&
    \multicolumn{2}{c |}{\textit{time (ms)}}&
    \multicolumn{2}{c}{\textit{predicates}}\\

    \midrule

    Without Contracts&
    9.67& 0\\

    Proxy only&
    132.67& 0\\

    \midrule
    Example~\ref{ex:addOne1}&
    58906&	300000&
    40383&  (- 31.44\%)&
    200000& (- 33.33\%)&
    40271&  (- 31.64\%)&
    200000& (- 33.33\%)\\

    Example~\ref{ex:addOne3}&
    131184& 600000&	
    87544&  (- 33.26\%)&
    400000& (- 33.33\%)&
    69492&  (- 47.03\%)&
    300000& (- 50.00\%)\\

    Example~\ref{ex:addOne4}&
    99776&	500000&
    81719&  (- 18.10\%)&
    400000& (- 20.00\%)&
    40218&  (- 59.69\%)&
    200000& (- 60.00\%)\\ 

    Example~\ref{ex:addOne4}&
    173037& 800000&
    128754& (- 25.59\%)&
    600000& (- 25.00\%)&
    66359&  (- 61.65\%)&
    300000& (- 62.50\%)\\

    Example~\ref{ex:addOne5}&
    221619& 1000000&
    160360& (- 27.64\%)&
    800000& (- 20.00\%)&
    88369&  (- 60.13\%)&
    500000& (- 50.00\%)\\

    Example~\ref{ex:addOne6}&
    441176& 1700000&
    278601&  (- 36.85\%)&
    1200000& (- 29.41\%)&
    163272&  (- 62.99\%)&
    700000&  (- 58.82\%)\\

    \bottomrule
  \end{tabular}

  \begin{tabular}{ l || r | r || r r | r r || r r | r r}
    \toprule
    
    \multicolumn{11}{c}{\textbf{\textit{No-Jit}}}\\

    \midrule

    \textbf{Benchmark}& 
    \multicolumn{2}{c ||}{\textbf{Normal}}&
    \multicolumn{4}{c ||}{\textbf{Baseline}}&
    \multicolumn{4}{c}{\textbf{Subset}}\\

    &
    \textit{time (ms)}&
    \textit{predicates}&
    \multicolumn{2}{c |}{\textit{time (ms)}}&
    \multicolumn{2}{c ||}{\textit{predicates}}&
    \multicolumn{2}{c |}{\textit{time (ms)}}&
    \multicolumn{2}{c}{\textit{predicates}}\\

    \midrule

    Without Contracts&
    40.33& 0\\

    Proxy only&
    200.67& 0\\

    \midrule

    Example~\ref{ex:addOne1}&
    81125&	300000&
    56439&  (- 30.43\%)&
    200000& (- 33.33\%)&
    54945&  (- 32.27\%)&
    200000& (- 33.33\%)\\

    Example~\ref{ex:addOne3}&
    186069& 600000&	
    124434& (- 33.12\%)&
    400000& (- 33.33\%)&
    96271&  (- 48.26\%)&
    300000& (- 50.00\%)\\

    Example~\ref{ex:addOne4}&
    136728&	500000&
    111596& (- 18.38\%)&
    400000& (- 20.00\%)&
    55186&  (- 59.64\%)&
    200000& (- 60.00\%)\\ 

    Example~\ref{ex:addOne4}&
    240724& 800000&
    179451& (- 25.45\%)&
    600000& (- 25.00\%)&
    91034&  (- 62.18\%)&
    300000& (- 62.50\%)\\

    Example~\ref{ex:addOne5}&
    315184& 1000000&
    225316& (- 28.51\%)&
    800000& (- 20.00\%)&
    123852& (- 60.71\%)&
    500000& (- 50.00\%)\\

    Example~\ref{ex:addOne6}&
    597406&  1700000&
    404276&  (- 32.33\%)&
    1200000& (- 29.41\%)&
    233124&  (- 60.98\%)&
    700000&  (- 58.82\%)\\

    \bottomrule
  \end{tabular}

  \caption{%
    Results from running the \TJS\ contract system.
    Table \textbf{\textit{Full}} shows the results of a run with both JIT
    compilers, whereas table \textbf{\textit{No-Ion}} shows the result of a run
    without IonMonkey (but the Baseline JIT remains enabled) and table
    \textbf{\textit{No-JIT}} shows the result of a run without any JIT
    compilation.
    Column \textbf{Normal} gives the baseline execution time and the total
    number of predicate checks of the unmodified program.
    Column \textbf{Baseline} and column \textbf{Subset} contain the execution
    time and the total number of predicate checks after applying the baseline
    simplification or the subset simplification, respectively.
    The value in parentheses indicates the improvement (in percent).
  }
  \label{fig:appendix/practical-evaluation/runtime}
\end{figure*}

Figure~\ref{fig:appendix/practical-evaluation/runtime} finally contains the
execution time and the number of predicate checks during execution of all
examples programs in different configuration, as explained in the figures
caption.

In addition to the examples from
Section~\ref{sec:appendix/practical-evaluation/examples}, it also contains the
run time of the \lstinline{addOne} program without contracts, and the run time
of the \lstinline{addOne} example where we only wrap the functions in a proxy,
but did not apply any contracts.

To sum up, the the \emph{Baseline Simplification} improves the run time by
approximately 28\%, whereas the \emph{Subset Simplification} makes an
improvement by approximately 58.68\%. Clearly, the improvement strictly depends
on the granularly of predicates and contracts.  

The numbers also indicate that the run time improvements remain identical,
whether we use an optimizing compiler or not. But, deactivating the optimizing
compiler increases the run time of the program without contracts by factor 4.83
(\textit{No-Ion}) and by factor 20.17 (\textit{No-Jit}), whereas the version with
contracts increases only by factor 1.5 (\textit{No-Ion}) and factor 2.09
(\textit{No-Jit}).

The numbers also shows that adding contracts, even without any functionality,
increases the run time by factor 71.2 (\textit{Full}), factor 13.7
(\textit{No-Ion}), and factor 5 (\textit{No-Jit}).

This illustrates that contracts prevent a program from being optimized
efficiently by an optimizing compiler.

\clearpage

\end{document}